\documentclass[draft]{article}
\usepackage{amsfonts,amsmath,amssymb,color,graphicx,latexsym}
\newtheorem{lemma}{Lemma}
\newtheorem{proposition}{Proposition}
\newtheorem{corollary}{Corollary}
\newenvironment{proof}{{\bf Proof:}}{~$\dashv$\\}
\def\N{\mathbb{N}}
\def\axiomN{\mathbf{N}}
\def\axiomwCD{\mathbf{wCD}}
\def\Axiom{\mathbf{A}}
\def\Rule{\mathbf{R}}
\def\LIK{\mathbf{LIK}}
\def\FIK{\mathbf{FIK}}
\def\Fo{\mathbf{Fo}}
\def\allfra{\mathbf{all}}
\def\fbcfra{\mathbf{fbc}}
\def\bcfra{\mathbf{bc}}
\def\fcfra{\mathbf{fc}}
\def\fbcfra{\mathbf{fbc}}
\def\fdcfra{\mathbf{fdc}}
\def\fucfra{\mathbf{fuc}}
\def\bdcfra{\mathbf{bdc}}
\def\bucfra{\mathbf{buc}}
\def\ducfra{\mathbf{duc}}
\def\fbdcfra{\mathbf{fbdc}}
\def\fbucfra{\mathbf{fbuc}}
\def\fducfra{\mathbf{fduc}}
\def\bducfra{\mathbf{bduc}}
\def\fbducfra{\mathbf{fbduc}}
\def\dcfra{\mathbf{dc}}
\def\ucfra{\mathbf{uc}}
\def\card{\mathtt{Card}}
\def\At{\mathbf{At}}
\def\K{\mathbf{K}}
\def\IK{\mathbf{IK}}
\def\S{\mathbf{S}}
\def\Log{\mathtt{Log}}
\def\IPL{\mathbf{IPL}}
\def\L{\mathbf{L}}
\def\hauteur{\mathtt{h}}
\def\degre{\mathtt{deg}}
\begin{document}
\title{Intuitionistic modal logics: new and simpler
\\
decidability proofs for $\FIK$ and $\LIK$}
\author{Philippe Balbiani\footnote{Corresponding author.
Email address: philippe.balbiani@irit.fr.}
\hspace{0.11cm}
\c{C}i\u gdem Gencer\footnote{Email addresses: cigdem.gencer@irit.fr.}}
\date{CNRS--INPT--UT3, IRIT, Toulouse, France}
\maketitle
\begin{abstract}
In this note, by integrating ideas concerning terminating tableaux-based procedures in modal logics and finite frame property of intuitionistic modal logic $\IK$, we provide new and simpler decidability proofs for $\FIK$ and $\LIK$.
\end{abstract}
{\bf Keywords:}
Intuitionistic modal logics.
Decidability and complexity.
\section{Introduction}\label{section:introduction}
Intuitionistic modal logic $\FIK$ has been introduced in~\cite{Balbiani:et:al:2024}.
Its semantics is based on Kripke frames $(W,{\leq},{R})$ where $\leq$ is a preorder on $W$ and $R$ is a binary relation on $W$ satisfying the {\em condition of forward confluence:} for all $s,t,u{\in}W$, if $t{\leq}s$ and $t{R}u$ then there exists $v{\in}W$ such that $s{R}v$ and $u{\leq}v$.
Intuitionistic modal logic $\LIK$ has been introduced in~\cite{Balbiani:et:al:2024b}.
Its semantics is based on Kripke frames $(W,{\leq},{R})$ where $\leq$ is a preorder on $W$ and $R$ is a binary relation on $W$ satisfying the condition of forward confluence and the {\em condition of downward confluence:} for all $s,t,u{\in}W$, if $s{\leq}t$ and $t{R}u$ then there exists $v{\in}W$ such that $s{R}v$ and $v{\leq}u$.
In~\cite{Balbiani:et:al:2024,Balbiani:et:al:2024b}, by integrating ideas coming from~\cite{Das:Marin:2023} about nested sequent calculi in intuitionistic modal logic $\IK$ and~\cite{Fitting:2014} about nested sequent calculi in first-order intuitionistic logic, Balbiani {\em et al.}\/ have demonstrated that the membership problems in $\FIK$ and $\LIK$ are decidable.
\\
\\
In~\cite{Gasquet:et:al:2006}, Gasquet {\em et al.}\/ have provided a general framework for defining tableau rules and tableau strategies in modal logics.
By using a simple graph rewriting language, they have given uniform proofs of soundness and completeness of multifarious terminating tableaux-based procedures in modal logics.
In~\cite{Grefe:1996}, Grefe has provided an algorithm constructing a finite countermodel for any formula not belonging to $\IK$.
As $\IK$ is finitely axiomatizable, he has therefore proved that the membership problem in $\IK$ is decidable.
The proofs of the decidability of the membership problems in $\FIK$ and $\LIK$ provided in~\cite{Balbiani:et:al:2024,Balbiani:et:al:2024b} are quite involved.
In this note, by integrating ideas coming from Gasquet {\em et al.}\/~\cite{Gasquet:et:al:2006} for what concerns terminating tableaux-based procedures in modal logics and Grefe~\cite{Grefe:1996} for what concerns finite frame property of $\IK$, we provide new and simpler decidability proofs for the membership problems in $\FIK$ and $\LIK$.
\\
\\
The usual approach to solving the membership problems in modal logics is to build finite models.
In modal logics, filtration is a fundamental concept for building finite models.
As far as we are aware, with a few exceptions such as~\cite{Hasimoto:2001,Sotirov:1984,Takano:2003}, it has not been so much adapted to intuitionistic modal logics, probably because it does not easily work with conditions involving the composition of binary relations such as the above-mentioned conditions of forward confluence and downward confluence.
In this note, considering that $\L$ is either $\FIK$, or $\LIK$, the new and simpler decidability proof for the membership problem in $\L$ that we provide is based on a selective filtration procedure: given a formula $A$, if $A{\not\in}\L$ then $A$ is falsified in the canonical model of $\L$ studied in Sections~\ref{section:canonical:model} and~\ref{section:maximal:worlds} and a selective filtration procedure defined in Sections~\ref{section:tips}--\ref{section:a:decision:procedure} can be used to constructing in Section~\ref{section:finite:model:property} a finite model falsifying $A$.
\section{Syntax}\label{section:closed:sets:of:formulas}
\subsection*{Notes on notation}
For all $n{\in}\N$, $(n)$ denotes $\{i{\in}\N:\ 1{\leq}i{\leq}n\}$.
\\
\\
For all sets $\Sigma$, $\card(\Sigma)$ denotes the {\em cardinal of $\Sigma$.}
\\
\\
For all sets $W$ and for all binary relations $R,S$ on $W$, ${R}{\circ}{S}$ denotes the {\em composition of $R$ and $S$,} i.e. the binary relation $T$ on $W$ such that for all $s,t{\in}W$, $s{T}t$ if and only if there exists $u{\in}W$ such that $s{R}u$ and $u{S}t$.
\\
\\
For all sets $W$ and for all binary relations $R$ on $W$, $R^{\star}$ denotes the least preorder on $W$ containing $R$.
\\
\\
For all sets $W$ and for all preorders $\leq$ on $W$, $\geq$ denotes the preorder on $W$ such that for all $s,t{\in}W$, $s{\geq}t$ if and only if $t{\leq}s$.
\\
\\
For all sets $W$ and for all preorders $\leq$ on $W$, a subset $U$ of $W$ is {\em $\leq$-closed}\/ if for all $s,t{\in}W$, if $s{\in}U$ and $s{\leq}t$ then $t{\in}U$.
\\
\\
For all sets $W$, for all preorders $\leq$ on $W$ and for all $s,t{\in}W$, when we write ``$s{<}t$'' we mean ``$s{\leq}t$ and $t{\not\leq}s$''.
\\
\\
Finally, ``$\IPL$'' stands for ``Intuitionistic Propositional Logic'' and ``IML'' stands for ``intuitionistic modal logic''.
\subsection*{Formulas}
Let $\At$ be a countably infinite set (with typical members called {\em atoms}\/ and denoted $p$, $q$, etc).
\\
\\
Let $\Fo$ be the countably infinite set (with typical members called {\em formulas}\/ and denoted $A$, $B$, etc) of finite words over $\At{\cup}\{{\rightarrow},{\top},{\bot},{\wedge},{\vee},{\square},{\lozenge},(,)\}$ defined by
$$A\ {::=}\ p{\mid}(A{\rightarrow}A){\mid}{\top}{\mid}{\bot}{\mid}(A{\wedge}A){\mid}(A{\vee}A){\mid}{\square}A{\mid}{\lozenge}A$$
where $p$ ranges over $\At$.
\\
\\
For all $A{\in}\Fo$, the {\em length of $A$}\/ (denoted ${\parallel}A{\parallel}$) is the number of symbols in $A$.
\\
\\
We follow the standard rules for omission of the parentheses.
\\
\\
For all $A{\in}\Fo$, when we write ``$\neg A$'' we mean ``$A{\rightarrow}{\bot}$''.
\subsection*{Sets of formulas}
Let ${\bowtie}$ be the binary relation between sets of formulas such that for all sets $\Gamma,\Delta$ of formulas, $\Gamma{\bowtie}\Delta$ if and only if for all $A{\in}\Fo$, the following conditions hold:
\begin{itemize}
\item if ${\square}A{\in}\Gamma$ then $A{\in}\Delta$,
\item if $A{\in}\Delta$ then ${\lozenge}A{\in}\Gamma$.
\end{itemize}
A set $\Sigma$ of formulas is {\em closed}\/ if for all $A,B{\in}\Fo$,
\begin{itemize}
\item if $A{\rightarrow}B{\in}\Sigma$ then $A{\in}\Sigma$ and $B{\in}\Sigma$,
\item if $A{\wedge}B{\in}\Sigma$ then $A{\in}\Sigma$ and $B{\in}\Sigma$,
\item if $A{\vee}B{\in}\Sigma$ then $A{\in}\Sigma$ and $B{\in}\Sigma$,
\item if ${\square}A{\in}\Sigma$ then $A{\in}\Sigma$,
\item if ${\lozenge}A{\in}\Sigma$ then $A{\in}\Sigma$.
\end{itemize}
For all $A{\in}\Fo$, let $\Sigma_{A}$ be the least closed set of formulas containing $A$.
\begin{lemma}
For all atoms $p$, $\Sigma_{p}{=}\{p\}$.
Moreover, for all $A,B{\in}\Fo$,
\begin{itemize}
\item $\Sigma_{A{\rightarrow}B}{=}\{A{\rightarrow}B\}{\cup}\Sigma_{A}{\cup}\Sigma_{B}$,
\item $\Sigma_{\top}{=}\{\top\}$,
\item $\Sigma_{\bot}{=}\{\bot\}$,
\item $\Sigma_{A{\wedge}B}{=}\{A{\wedge}B\}{\cup}\Sigma_{A}{\cup}\Sigma_{B}$,
\item $\Sigma_{A{\vee}B}{=}\{A{\vee}B\}{\cup}\Sigma_{A}{\cup}\Sigma_{B}$,
\item $\Sigma_{{\square}A}{=}\{{\square}A\}{\cup}\Sigma_{A}$,
\item $\Sigma_{{\lozenge}A}{=}\{{\lozenge}A\}{\cup}\Sigma_{A}$.
\end{itemize}
\end{lemma}
\begin{proof}
Left to the reader.
\medskip
\end{proof}
\begin{lemma}
For all $A{\in}\Fo$, $\Sigma_{A}$ is finite.
More precisely, for all $A{\in}\Fo$, $\card(\Sigma_{A})
$\linebreak$
{\leq}{\parallel}A{\parallel}$.
\end{lemma}
\begin{proof}
By induction on $A$.
\medskip
\end{proof}
For all $A{\in}\Fo$ and for all $\Gamma{\subseteq}\Sigma_{A}$, let $\Gamma^{\circ}{=}{\bigcup}\{\Sigma_{B}$: $B{\in}\Fo$ is such that either ${\square}B{\in}\Gamma$, or ${\lozenge}B{\in}\Gamma\}$.
\begin{lemma}
For all $A{\in}\Fo$ and for all $\Gamma{\subseteq}\Sigma_{A}$, $\Gamma^{\circ}$ is closed.
Moreover, if $\Gamma$ is closed then $\Gamma^{\circ}{\subseteq}\Gamma$.
If, in addition, $\Gamma{\not=}\emptyset$ then $\card(\Gamma^{\circ}){<}\card(\Gamma)$.
\end{lemma}
\begin{proof}
Left to the reader.
\medskip
\end{proof}
For all $A{\in}\Fo$, for all $\Gamma{\subseteq}\Sigma_{A}$ and for all $\alpha{\in}\N$, let $\Gamma^{\alpha}$ be inductively defined as follows:
\begin{itemize}
\item $\Gamma^{0}{=}\Gamma$,
\item $\Gamma^{\alpha{+}1}{=}{\Gamma^{\alpha}}^{\circ}$.
\end{itemize}
\begin{lemma}
For all $A{\in}\Fo$, for all $\Gamma{\subseteq}\Sigma_{A}$ and for all $\alpha{\in}\N$, $\card(\Gamma^{\alpha}){\leq}\max\{0,
$\linebreak$
\card(\Sigma_{A}){-}\alpha\}$.
\end{lemma}
\begin{proof}
By induction on $\alpha$.
\medskip
\end{proof}
\section{Semantics}\label{section:semantics}
\subsection*{Frames}
A {\em frame}\/ is a relational structure of the form $(W,{\leq},{R})$ where $W$ is a nonempty set of {\em worlds,} $\leq$ is a preorder on $W$ and ${R}$ is a binary relation on $W$.
\\
\\
Let ${\mathcal C}_{\allfra}$ be the class of all frames.
\\
\\
For all classes ${\mathcal C}$ of frames, let ${\mathcal C}^{\omega}$ be the set of all finite frames in ${\mathcal C}$.
\subsection*{Confluences}
A frame $(W,{\leq},{R})$ is
\begin{itemize}
\item {\em forward confluent}\/ if ${\geq}{\circ}{R}{\subseteq}{R}{\circ}{\geq}$, i.e. for all $s,t,u{\in}W$, if $t{\leq}s$ and $t{R}u$ then there exists $v{\in}W$ such that $s{R}v$ and $u{\leq}v$,
\item {\em backward confluent}\/ if ${R}{\circ}{\leq}{\subseteq}{\leq}{\circ}{R}$, i.e. for all $s,t,u{\in}W$, if $s{R}t$ and $t{\leq}u$ then there exists $v{\in}W$ such that $s{\leq}v$ and $v{R}u$,
\item {\em downward confluent}\/ if ${\leq}{\circ}{R}{\subseteq}{R}{\circ}{\leq}$, i.e. for all $s,t,u{\in}W$, if $s{\leq}t$ and $t{R}u$ then there exists $v{\in}W$ such that $s{R}v$ and $v{\leq}u$,
\item {\em upward confluent}\/ if ${R}{\circ}{\geq}{\subseteq}{\geq}{\circ}{R}$, i.e. for all $s,t,u{\in}W$, if $s{R}t$ and $u{\leq}t$ then there exists $v{\in}W$ such that $v{\leq}s$ and $v{R}u$.
\end{itemize}
Let ${\mathcal C}_{\fcfra}$ be the class of all forward confluent frames, ${\mathcal C}_{\bcfra}$ be the class of all backward confluent frames, ${\mathcal C}_{\dcfra}$ be the class of all downward confluent frames, ${\mathcal C}_{\ucfra}$ be the class of all upward confluent frames, ${\mathcal C}_{\fbcfra}$ be the class of all forward confluent and backward confluent frames, ${\mathcal C}_{\fdcfra}$ be the class of all forward confluent and downward confluent frames, ${\mathcal C}_{\fucfra}$ be the class of all forward confluent and upward confluent frames, ${\mathcal C}_{\bdcfra}$ be the class of all backward confluent and downward confluent frames, ${\mathcal C}_{\bucfra}$ be the class of all backward confluent and upward confluent frames, ${\mathcal C}_{\ducfra}$ be the class of all downward confluent and upward confluent frames, ${\mathcal C}_{\fbdcfra}$ be the class of all forward confluent, backward confluent and downward confluent frames, ${\mathcal C}_{\fbucfra}$ be the class of all forward confluent, backward confluent and upward confluent frames, ${\mathcal C}_{\fducfra}$ be the class of all forward confluent, downward confluent and upward confluent frames, ${\mathcal C}_{\bducfra}$ be the class of all backward confluent, downward confluent and upward confluent frames and ${\mathcal C}_{\fbducfra}$ be the class of all forward confluent, backward confluent, downward confluent and upward confluent frames.
\subsection*{Valuations}
A {\em valuation on a frame $(W,{\leq},{R})$}\/ is a function $V\ :\ \At\longrightarrow\wp(W)$ such that for all atoms $p$, $V(p)$ is $\leq$-closed.
\subsection*{Models}
A {\em model}\/ is a $4$-tuple consisting of the $3$ components of a frame and a valuation on that frame.
\subsection*{Satisfiability}
With respect to a model $(W,{\leq},{R},V)$, for all $s{\in}W$ and for all $A{\in}\Fo$, the {\em satisfiability of $A$ at $s$ in $(W,{\leq},{R},V)$}\/ (in symbols $(W,{\leq},{R},V),s{\models}A$) is inductively defined as follows:
\begin{itemize}
\item $(W,{\leq},{R},V),s{\models}p$ if and only if $s{\in}V(p)$,
\item $(W,{\leq},{R},V),s{\models}A{\rightarrow}B$ if and only if for all $t{\in}W$, if $s{\leq}t$ and $(W,{\leq},{R},V),t
$\linebreak$
{\models}A$ then $(W,{\leq},{R},V),t{\models}B$,
\item $(W,{\leq},{R},V),s{\models}{\top}$,
\item $(W,{\leq},{R},V),s{\not\models}{\bot}$,
\item $(W,{\leq},{R},V),s{\models}A{\wedge}B$ if and only if $(W,{\leq},{R},V),s{\models}A$ and $(W,{\leq},{R},V),s
$\linebreak$
{\models}B$,
\item $(W,{\leq},{R},V),s{\models}A{\vee}B$ if and only if either $(W,{\leq},{R},V),s{\models}A$, or $(W,{\leq},{R},
$\linebreak$
V),s{\models}B$,
\item $(W,{\leq},{R},V),s{\models}{\square}A$ if and only if for all $t{\in}W$, if $s{\leq}t$ then for all $u{\in}W$, if $t{R}u$ then $(W,{\leq},{R},V),u{\models}A$,
\item $(W,{\leq},{R},V),s{\models}{\lozenge}A$ if and only if there exists $t{\in}W$ such that $t{\leq}s$ and there exists $u{\in}W$ such that $t{R}u$ and $(W,{\leq},{R},V),u{\models}A$.
\end{itemize}
For all models $(W,{\leq},{R},V)$, for all $s{\in}W$ and for all $A{\in}\Fo$, if $(W,{\leq},{R},V)$ is clear from the context then when we write ``$s{\models}A$'' we mean ``$(W,{\leq},{R},V),s{\models}A$''.
\begin{lemma}[Heredity Property]
Let $(W,{\leq},{R},V)$ be a model.
For all $A{\in}\Fo$ and for all $s,t{\in}W$, if $s{\models}A$ and $s{\leq}t$ then $t{\models}A$.
\end{lemma}
\begin{proof}
By induction on $A$.
\medskip
\end{proof}
\subsection*{About Fischer Servi and Wijesekera}
Fischer Servi~\cite{FischerServi:1984} has defined the satisfiability of $\lozenge$-formulas as follows:
\begin{itemize}
\item $(W,{\leq},{R},V),s{\models_{\mathtt{FS}}}{\lozenge}A$ if and only if there exists $t{\in}W$ such that $s{R}t$ and $(W,{\leq},{R},V),t{\models_{\mathtt{FS}}}A$.
\end{itemize}
The definition of the satisfiability of formulas considered by Fischer Servi necessitates to restrict the discussion to the class of all forward confluent frames, otherwise the above-described Heredity Property would not hold.
\\
\\
Wijesekera~\cite{Wijesekera:1990} has defined the satisfiability of $\lozenge$-formulas as follows:
\begin{itemize}
\item $(W,{\leq},{R},V),s{\models_{\mathtt{W}}}{\lozenge}A$ if and only if for all $t{\in}W$, if $s{\leq}t$ then there exists $u{\in}W$ such that $t{R}u$ and $(W,{\leq},{R},V),u{\models_{\mathtt{W}}}A$.
\end{itemize}
The definition of the satisfiability of formulas considered by Wijesekera does not necessitate to restrict the discussion to a specific class of frames.
\begin{lemma}
In the class of all forward confluent frames, the definition of the satisfiability of formulas considered by Fischer Servi, the definition of the satisfiability of formulas considered by Wijesekera and our definition of the satisfiability of formulas give rise to the same relation of satisfiability between formulas, worlds and models.
\end{lemma}
\begin{proof}
Let $(W,{\leq},{R},V)$ be a model.
Suppose $(W,{\leq},{R})$ is forward confluent.
By induction on $A$, the reader may easily verify that for all $s{\in}W$, $s{\models_{\mathtt{FS}}}A$ if and only if $s{\models_{\mathtt{W}}}A$ if and only if $s{\models}A$.
\medskip
\end{proof}
\subsection*{Validity}
A formula $A$ is {\em true in a model $(W,{\leq},{R},V)$}\/ (in symbols $(W,{\leq},{R},V){\models}A$) if for all $s{\in}W$, $s{\models}A$.
\\
\\
A formula $A$ is {\em valid in a frame $(W,{\leq},{R})$}\/ (in symbols $(W,{\leq},{R}){\models}A$) if for all models $(W,{\leq},{R},V)$ based on $(W,{\leq},{R})$, $(W,{\leq},{R},V){\models}A$.
\\
\\
For all classes ${\mathcal C}$ of frames, let $\Log({\mathcal C}){=}\{A{\in}\Fo\ :\ \text{for}\ \text{all}\ (W,{\leq},{R}){\in}{\mathcal C},\ (W,{\leq},{R}){\models}
$\linebreak$
A\}$.
\begin{proposition}\label{proposition:facile:about:inclusion}
\begin{itemize}
\item $\Log({\mathcal C}_{\fcfra}){\subseteq}\Log({\mathcal C}_{\fucfra}){\subseteq}\Log({\mathcal C}_{\fucfra}^{\omega})$,
\item $\Log({\mathcal C}_{\fdcfra}){\subseteq}\Log({\mathcal C}_{\fducfra}){\subseteq}\Log({\mathcal C}_{\fducfra}^{\omega})$.
\end{itemize}
\end{proposition}
\begin{proof}
Left to the reader.
\medskip
\end{proof}
\section{Axiomatization}\label{section:axiomatization}
\subsection*{Preliminaries}
See~\cite[Chapter~$2$]{Chagrov:Zakharyaschev:1997} for an introduction to the standard axioms of $\IPL$ and the standard inference rules of $\IPL$.
\\
\\
The axioms
\begin{description}
%
%
%
%
%
%
\item[$(\mathbf{D}\square)$] ${\square}p{\wedge}{\square}q{\rightarrow}{\square}(p{\wedge}q)$,
\item[$(\mathbf{D}\lozenge)$] ${\lozenge}(p{\vee}q){\rightarrow}{\lozenge}p{\vee}{\lozenge}q$,
\item[$(\axiomN{\square})$] ${\square}{\top}$,
\item[$(\axiomN{\lozenge})$] ${\neg}{\lozenge}{\bot}$,
\item[$(\axiomwCD)$] ${\square}(p{\vee}q){\rightarrow}(({\lozenge}p{\rightarrow}{\square}q){\rightarrow}{\square}q)$,
\item[$(\Axiom\mathbf{f})$] ${\lozenge}(p{\rightarrow}q){\rightarrow}({\square}p{\rightarrow}{\lozenge}q)$,
\item[$(\Axiom\mathbf{b})$] $({\lozenge}p{\rightarrow}{\square}q){\rightarrow}{\square}(p{\rightarrow}q)$,
\item[$(\Axiom\mathbf{d})$] ${\square}(p{\vee}q){\rightarrow}{\lozenge}p{\vee}{\square}q$,
\end{description}
and the inference rules
\begin{description}
\item[$(\Rule{\square})$] $\frac{p{\rightarrow}q}{{\square}p{\rightarrow}{\square}q}$,
\item[$(\Rule{\lozenge})$] $\frac{p{\rightarrow}q}{{\lozenge}p{\rightarrow}{\lozenge}q}$,
\item[$(\Rule\mathbf{I})$] $\frac{{\lozenge}p{\rightarrow}q{\vee}{\square}(p{\rightarrow}r)}{{\lozenge}p{\rightarrow}q{\vee}{\lozenge}r}$,
\end{description}
have been recently used in IMLs~\cite{Balbiani:et:al:2021,Balbiani:et:al:2024,Balbiani:et:al:2024b,Balbiani:Gencer:preliminary:draft,Dalmonte:et:al:2021,Das:Marin:2023,VanDerGiessen:2023,Girlando:et:al:2023,Girlando:et:al:2024,Marin:et:al:2021,Mendler:et:al:2021,Olivetti:2022}.
\subsection*{Intuitionistic modal logics}
An {\em IML}\/ is a set of formulas closed under uniform substitution, containing the standard axioms of $\IPL$, closed under the standard inference rules of $\IPL$, containing the axioms $(\mathbf{D}\square)$, $(\mathbf{D}\lozenge)$, $(\axiomN{\square})$, $(\axiomN{\lozenge})$ and$(\axiomwCD)$ and closed under the inference rules $(\Rule{\square})$, $(\Rule{\lozenge})$ and $(\Rule\mathbf{I})$.
\\
\\
The reader may easily see that axiom $(\axiomwCD)$ has similarities with P\v{r}enosil's equation ${\lozenge}a{\rightarrow}{\square}b{\leq}{\square}(a{\vee}b){\rightarrow}{\square}b$ and inference rule $(\Rule\mathbf{I})$ has similarities with P\v{r}enosil's {\em positive modal law}\/ ${\lozenge}b{\leq}{\square}a{\vee}c{\Rightarrow}{\lozenge}b{\leq}{\lozenge}(a{\wedge}b){\vee}c$~\cite{Prenosil:2014}.
\\
\\
Let $\L_{\min}$ be the least IML.
\\
\\
For all IMLs $\L$ and for all sets $\Sigma$ of formulas, let $\L{\oplus}\Sigma$ be the least IML containing $\L$ and $\Sigma$.
\\
\\
For all IMLs $\L$ and for all formulas $A$, we write $\L{\oplus}A$ instead of $\L{\oplus}\{A\}$.
\\
\\
Let $\L_{\fcfra}{=}\L_{\min}{\oplus}(\Axiom\mathbf{f})$, $\L_{\bcfra}{=}\L_{\min}{\oplus}(\Axiom\mathbf{b})$ and $\L_{\dcfra}{=}\L_{\min}{\oplus}(\Axiom\mathbf{d})$.
We also write $\L_{\fbcfra}$ to denote $\L_{\min}{\oplus}\{(\Axiom\mathbf{f}),(\Axiom\mathbf{b})\}$, $\L_{\fdcfra}$ to denote $\L_{\min}{\oplus}\{(\Axiom\mathbf{f}),(\Axiom\mathbf{d})\}$, etc.
%
%
%
%
%
%
%
%
%
%
%
%
%
%
%
%
%
%
%
%
%
%
%
%
%
%
%
%
%
%
%
%
%
%
%
%
%
%
%
%
%
%
%
%
%
%
%
%
%
%
%
%
%
%
\section{Canonical model}\label{section:canonical:model}
From now on in this note, let $\L$ be either $\FIK$, or $\LIK$.
\subsection*{Theories}
A {\em theory}\/ is a set of formulas containing $\L$ and closed with respect to the inference rule of modus ponens.
\begin{lemma}\label{lemma:L:is:theory}
$\L$ is a theory.
\end{lemma}
\begin{proof}
Left to the reader.
\medskip
\end{proof}
A theory $s$ is {\em proper}\/ if ${\bot}{\not\in}s$.
\begin{lemma}\label{lemma:L:is:proper:theory}
$\L$ is proper.
\end{lemma}
\begin{proof}
By Proposition~\ref{proposition:soundness:FIK:LIK}.
\medskip
\end{proof}
A proper theory $s$ is {\em prime}\/ if for all $A,B{\in}\Fo$, if $A{\vee}B{\in}s$ then either $A{\in}s$, or $B{\in}s$.
\begin{lemma}
There exists prime theories.
\end{lemma}
\begin{proof}
By Lindenbaum Lemma~\cite{Balbiani:Gencer:preliminary:draft} and Lemma~\ref{lemma:L:is:proper:theory}.
\medskip
\end{proof}
\subsection*{The canonical frame}
Let $(W_{\L},{\leq_{\L}},{R_{\L}})$ be the frame such that
\begin{itemize}
\item $W_{\L}$ is the nonempty set of all prime theories,
\item $\leq_{\L}$ is the preorder on $W_{\L}$ such that for all $s,t{\in}W_{\L}$, $s{\leq_{\L}}t$ if and only if $s{\subseteq}t$,
\item $R_{\L}$ is the binary relation on $W_{\L}$ such that for all $s,t{\in}W_{\L}$, $s{R_{\L}}t$ if and only if $s{\bowtie}t$.
\end{itemize}
The frame $(W_{\L},{\leq_{\L}},{R_{\L}})$ is called {\em canonical frame of $\L$.}
\begin{lemma}\label{lemma:canonical:frame:confluences}
\begin{itemize}
\item $(W_{\L},{\leq_{\L}},{R_{\L}})$ is in ${\mathcal C}_{\fcfra}$ when $\L{=}\FIK$,
\item $(W_{\L},{\leq_{\L}},{R_{\L}})$ is in ${\mathcal C}_{\fdcfra}$ when $\L{=}\LIK$.
\end{itemize}
\end{lemma}
\begin{proof}
See~\cite{Balbiani:Gencer:preliminary:draft}.
\medskip
\end{proof}
\subsection*{The canonical valuation}
The {\em canonical valuation of $\L$}\/ is the valuation $V_{\L}\ :\ \At\longrightarrow\wp(W_{\L})$ on $(W_{\L},{\leq_{\L}},
$\linebreak$
{R_{\L}})$ such that for all atoms $p$, $V_{\L}(p){=}\{s{\in}W_{\L}\ :\ p{\in}s\}$.
\subsection*{The canonical model}
The {\em canonical model of $\L$}\/ is the model $(W_{\L},{\leq_{\L}},{R_{\L}},V_{\L})$.
\begin{lemma}[Canonical Truth Lemma]\label{lemma:truth:lemma}
For all $A{\in}\Fo$ and for all $s{\in}W_{\L}$, $A{\in}s$ if and only if $(W_{\L},{\leq_{\L}},{R_{\L}},V_{\L}),s{\models}A$.
\end{lemma}
\begin{proof}
See~\cite{Balbiani:Gencer:preliminary:draft}.
\medskip
\end{proof}
\begin{lemma}\label{lemma:proof:of:completeness}
For all $A{\in}\Fo$,
\begin{itemize}
\item if $A{\not\in}\L$ then ${\mathcal C}_{\fcfra}{\not\models}A$ when $\L{=}\FIK$,
\item if $A{\not\in}\L$ then ${\mathcal C}_{\fdcfra}{\not\models}A$ when $\L{=}\LIK$.
\end{itemize}
\end{lemma}
\begin{proof}
By Lindenbaum Lemma~\cite{Balbiani:Gencer:preliminary:draft} and Lemmas~\ref{lemma:L:is:theory}, \ref{lemma:canonical:frame:confluences} and~\ref{lemma:truth:lemma}.
\medskip
\end{proof}
\begin{proposition}[Completeness]\label{proposition:completeness:of:FIK:and:LIK}
\begin{itemize}
\item $\FIK{=}\Log({\mathcal C}_{\fcfra})$,
\item $\LIK{=}\Log({\mathcal C}_{\fdcfra})$.
\end{itemize}
\end{proposition}
\begin{proof}
By Proposition~\ref{proposition:soundness:FIK:LIK} and Lemma~\ref{lemma:proof:of:completeness}.
\medskip
\end{proof}
\section{Maximal worlds}\label{section:maximal:worlds}
We say that $s{\in}W_{\L}$ is {\em maximal with respect to $B{\in}\Fo$}\/ if for all $t{\in}W_{\L}$, if $s{<_{\L}}t$ then $t{\models}B$.
\begin{lemma}\label{lemma:maximal}
Let $s{\in}W_{\L}$ and $B{\in}\Fo$.
If $s$ is not maximal with respect to $B$ then there exists $t{\in}W_{\L}$ such that $s{<_{\L}}t$, $t{\not\models}B$ and $t$ is maximal with respect to $B$.
\end{lemma}
\begin{proof}
Suppose $s$ is not maximal with respect to $B$.
Let $S=\{t{\in}W_{\L}$: $s{<_{\L}}t$ and $t{\not\models}B\}$.
Since $s$ is not maximal with respect to $B$, then $S$ is nonempty.
Moreover, obviously, for all nonempty chains $(t_{i})_{i{\in}I}$ of elements of $S$, $\bigcup\{t_{i}$: $i{\in}I\}$ is in $S$.
Hence, by Zorn's Lemma, $S$ possesses a maximal element $t$.
Obviously, $s{<_{\L}}t$, $t{\not\models}B$ and $t$ is maximal with respect to $B$.
\medskip
\end{proof}
\begin{lemma}\label{lemma:maximal:rightarrow}
Let $s{\in}W_{\L}$ and $B,C{\in}\Fo$.
If $s{\not\models}B{\rightarrow}C$ and $s$ is maximal with respect to $B{\rightarrow}C$ then $s{\models}B$ and $s{\not\models}C$.
\end{lemma}
\begin{proof}
Suppose $s{\not\models}B{\rightarrow}C$ and $s$ is maximal with respect to $B{\rightarrow}C$.
For the sake of the contradiction, suppose either $s{\not\models}B$, or $s{\models}C$.
Since $s{\not\models}B{\rightarrow}C$, then there exists $t{\in}W_{\L}$ such that $s{\leq_{\L}}t$, $t{\models}B$ and $t{\not\models}C$.
Hence, either $s{=}t$, or $s{<_{\L}}t$.
In the former case, since $t{\models}B$ and $t{\not\models}C$, then $s{\models}B$ and $s{\not\models}C$.
Since either $s{\not\models}B$, or $s{\models}C$, then $s{\models}C$: a contradiction.
In the latter case, since $s$ is maximal with respect to $B{\rightarrow}C$, then $t{\models}B{\rightarrow}C$.
Since $t{\models}B$, then $t{\models}C$: a contradiction.
\medskip
\end{proof}
\begin{lemma}\label{lemma:maximal:square}
Let $s{\in}W_{\L}$ and $B{\in}\Fo$.
If $s{\not\models}{\square}B$ and $s$ is maximal with respect to ${\square}B$ then there exists $t{\in}W_{\L}$ such that $t{\in}R_{\L}(s)$ and $t{\not\models}B$.
\end{lemma}
\begin{proof}
Suppose $s{\not\models}{\square}B$ and $s$ is maximal with respect to ${\square}B$.
Hence, there exists $t{\in}W_{\L}$ such that $s{\leq_{\L}}t$ and there exists $u{\in}W_{\L}$ such that $t{R_{\L}}u$ and $u{\not\models}B$.
Thus, either $s{=}t$, or $s{<_{\L}}t$.
In the former case, since $t{R_{\L}}u$, then $u{\in}R_{\L}(s)$.
In the latter case, since $s$ is maximal with respect to ${\square}B$, then $t{\models}{\square}B$.
Since $t{R_{\L}}u$, then $u{\models}B$: a contradiction.
\medskip
\end{proof}
\begin{lemma}\label{lemma:lozenge}
Let $s{\in}W_{\L}$ and $B{\in}\Fo$.
If $s{\models}{\lozenge}B$ then there exists $t{\in}W_{\L}$ such that $t{\in}R_{\L}(s)$ and $t{\models}B$.
\end{lemma}
\begin{proof}
Suppose $s{\models}{\lozenge}B$.
Hence, there exists $t{\in}W_{\L}$ such that $t{\leq_{\L}}s$ and there exists $u{\in}W_{\L}$ such that $t{R_{\L}}u$ and $u{\models}B$.
Since $(W_{\L},{\leq_{\L}},{R_{\L}})$ is in ${\mathcal C}_{\fcfra}$ when $\L{=}\FIK$ and in ${\mathcal C}_{\fdcfra}$ when $\L{=}\LIK$, then $(W_{\L},{\leq_{\L}},{R_{\L}})$ is forward confluent.
Since $t{\leq_{\L}}s$ and $t{R_{\L}}u$, then there exists $v{\in}W_{\L}$ such that $s{R_{\L}}v$ and $u{\leq_{\L}}v$.
Thus, $v{\in}R_{\L}(s)$.
Moreover, since $u{\models}B$, then $v{\models}B$.
\medskip
\end{proof}
\section{Tips and clips}\label{section:tips}
From now on in this note, let $A{\in}\Fo$ be such that $A{\not\in}\L$.
\\
\\
Therefore, by Lindenbaum Lemma~\cite{Balbiani:Gencer:preliminary:draft} and Lemmas~\ref{lemma:L:is:theory} and~\ref{lemma:truth:lemma}, there exists $s_{0}{\in}W_{\L}$ such that $(W_{\L},{\leq_{\L}},{R_{\L}},V_{\L}),s_{0}{\not\models}A$.
\subsection*{Tips}
A {\em tip}\/ is a $5$-tuple of the form $(i,s,\Gamma,\alpha,X)$ where $i{\in}\N$, $s{\in}W_{\L}$, $\Gamma$ is a closed subset of $\Sigma_{A}$, $\alpha{\in}\N$ and $X{\in}\N$.
\\
\\
The $5$-tuple $(0,s_{0},\Sigma_{A},0,0)$ is called {\em initial tip of $s_{0}$.}
\\
\\
The {\em name of the tip $(i,s,\Gamma,\alpha,X)$}\/ is $i$.
\\
\\
The {\em correspondent of the tip $(i,s,\Gamma,\alpha,X)$}\/ is $s$.
\\
\\
The {\em topic of the tip $(i,s,\Gamma,\alpha,X)$}\/ is $\Gamma$.
\\
\\
The {\em rank of the tip $(i,s,\Gamma,\alpha,X)$}\/ is $\alpha$.
\\
\\
The {\em height of the tip $(i,s,\Gamma,\alpha,X)$}\/ is $X$.
\\
\\
The {\em degree of the tip $(i,s,\Gamma,\alpha,X)$}\/ (denoted $\degre(i,s,\Gamma,\alpha,X)$) is the number of $B{\in}\Gamma$ such that $s{\not\models}B$ and $s$ is not maximal with respect to $B$.
\begin{lemma}
$\degre(0,s_{0},\Sigma_{A},0,0){\leq}\card(\Sigma_{A})$.
\end{lemma}
\begin{proof}
Left to the reader.
\medskip
\end{proof}
For all finite nonempty sets ${\mathcal T}$ of tips and for all $\alpha{\in}\N$, let $\hauteur_{\alpha}({\mathcal T}){=}\max\{X{\in}\N$: $(i,s,\Gamma,\alpha,X){\in}{\mathcal T}\}$.
\subsection*{Clips}
A {\em clip}\/ is a triple of the form $({\mathcal T},{\ll},{\triangleright})$ where ${\mathcal T}$ is a finite nonempty set of tips and $\ll$ and $\triangleright$ are binary relations on $\mathcal T$.
\\
\\
The triple $(\{(0,s_{0},\Sigma_{A},0,0)\},\emptyset,\emptyset)$ is called {\em initial clip of $s_{0}$.}
\subsection*{Coherence}
We say that the clip $({\mathcal T},{\ll},{\triangleright})$ is {\em coherent}\/ if for all $(i,s,\Gamma,\alpha,X),(j,
$\linebreak$
t,\Delta,\beta,Y){\in}{\mathcal T}$,
\begin{itemize}
\item if $i{=}j$ then $s{=}t$, $\Gamma{=}\Delta$, $\alpha{=}\beta$ and $X{=}Y$,
\item if $(i,s,\Gamma,\alpha,X){\ll}(j,t,\Delta,\beta,Y)$ then $j{\not=}i$, $s{\leq_{\L}}t$, $\Gamma{\not=}\emptyset$, $\Delta{=}\Gamma$, $\beta{=}\alpha$ and $Y{=}X
$\linebreak$
{+}1$,
\item if $(i,s,\Gamma,\alpha,X){\triangleright}(j,t,\Delta,\beta,Y)$ then $j{\not=}i$, $t{\in}R_{\L}(s)$, $\Gamma{\not=}\emptyset$, $\Delta{=}\Gamma^{\circ}$, $\beta{=}\alpha{+}1$ and $Y{=}X$.
\end{itemize}
\begin{lemma}
$(\{(0,s_{0},\Sigma_{A},0,0)\},\emptyset,\emptyset)$ is coherent.
\end{lemma}
\begin{proof}
Left to the reader.
\medskip
\end{proof}
\begin{lemma}\label{lemma:about:coherence:and:morphisms}
Let $({\mathcal T},{\ll},{\triangleright})$ be a clip.
If $({\mathcal T},{\ll},{\triangleright})$ is coherent then the function $f:\ {\mathcal T}{\longrightarrow}W_{\L}$ such that for all $(i,s,\Gamma,\alpha,X){\in}{\mathcal T}$, $f(i,s,\Gamma,\alpha,X){=}s$ is a homomorphism from $({\mathcal T},{\ll^{\star}},{\triangleright})$ to $(W_{\L},{\leq_{\L}},{R_{\L}})$.
\end{lemma}
\begin{proof}
Suppose $({\mathcal T},{\ll},{\triangleright})$ is coherent.
Let $(i,s,\Gamma,\alpha,X),(j,t,\Delta,\beta,Y){\in}{\mathcal T}$.
\linebreak
Firstly, suppose $(i,s,\Gamma,\alpha,X){\ll^{\star}}(j,t,\Delta,\beta,Y)$.
Hence, there exists $n{\in}\N$ and there exists $(k_{0},u_{0},\Lambda_{0},\gamma_{0},Z_{0}),\ldots,(k_{n},u_{n},\Lambda_{n},\gamma_{n},Z_{n}){\in}{\mathcal T}$ such that $(k_{0},u_{0},\Lambda_{0},
$\linebreak$
\gamma_{0},Z_{0}){=}(i,s,\Gamma,\alpha,X)$, $(k_{n},u_{n},\Lambda_{n},\gamma_{n},Z_{n}){=}(j,t,\Delta,\beta,Y)$ and for all $i{\in}(n)$,
\linebreak$
(k_{i{-}1},u_{i{-}1},\Lambda_{i{-}1},\gamma_{i{-}1},Z_{i{-}1}){\ll}(k_{i},u_{i},\Lambda_{i},\gamma_{i},Z_{i})$.
By induction on $n$, the reader may easily verify that $s{\leq_{\L}}t$.
Secondly, suppose $(i,s,\Gamma,\alpha,X){\triangleright}(j,t,\Delta,\beta,Y)$.
Since ${\mathcal T}$ is coherent, then $s{R_{\L}}t$.
\medskip
\end{proof}
\subsection*{Regularity}
We say that the clip $({\mathcal T},{\ll},{\triangleright})$ is {\em regular}\/ if for all $(i,s,\Gamma,\alpha,X),(j,t,\Delta,\beta,Y),(k,
$\linebreak$
u,\Lambda,\gamma,Z){\in}{\mathcal T}$,
\begin{itemize}
\item if $(i,s,\Gamma,\alpha,X){\ll}(k,u,\Lambda,\gamma,Z)$ and $(j,t,\Delta,\beta,Y){\ll}(k,u,\Lambda,\gamma,Z)$ then $i{=}j$,
\item if $(i,s,\Gamma,\alpha,X){\triangleright}(k,u,\Lambda,\gamma,Z)$ and $(j,t,\Delta,\beta,Y){\triangleright}(k,u,\Lambda,\gamma,Z)$ then $i{=}j$,
\item if $(i,s,\Gamma,\alpha,X){\triangleright}(k,u,\Lambda,\gamma,Z)$ and $(j,t,\Delta,\beta,Y){\ll}(k,u,\Lambda,\gamma,Z)$ then there exists $(l,v,\Theta,\delta,T){\in}{\mathcal T}$ such that $(l,v,\Theta,\delta,T){\ll}(i,s,\Gamma,\alpha,X)$ and $(l,v,\Theta,\delta,
$\linebreak$
T){\triangleright}(j,t,\Delta,\beta,Y)$.
\end{itemize}
\begin{lemma}
$(\{(0,s_{0},\Sigma_{A},0,0)\},\emptyset,\emptyset)$ is regular.
\end{lemma}
\begin{proof}
Left to the reader.
\medskip
\end{proof}
\begin{lemma}\label{lemma:about:regularity:and:upward:confluence}
Let $({\mathcal T},{\ll},{\triangleright})$ be a clip.
If $({\mathcal T},{\ll},{\triangleright})$ is regular then $({\mathcal T},{\ll^{\star}},{\triangleright})$ is upward confluent.
\end{lemma}
\begin{proof}
Suppose $({\mathcal T},{\ll},{\triangleright})$ is regular.
Let $(i,s,\Gamma,\alpha,X),(j,t,\Delta,\beta,Y),(k,u,\Lambda,\gamma,
$\linebreak$
Z){\in}{\mathcal T}$.
Suppose $(i,s,\Gamma,\alpha,X){\triangleright}(j,t,\Delta,\beta,Y)$ and $(k,u,\Lambda,\gamma,Z){\ll^{\star}}(j,t,\Delta,\beta,Y)$.
Hence, there exists $n{\in}\N$ and there exists $(l_{0},v_{0},\Theta_{0},\delta_{0},T_{0}),\ldots,(l_{n},v_{n},\Theta_{n},\delta_{n},
$\linebreak$
T_{n}){\in}{\mathcal T}$ such that $(l_{0},v_{0},\Theta_{0},\delta_{0},T_{0}){=}(k,u,\Lambda,\gamma,Z)$, $(l_{n},v_{n},\Theta_{n},\delta_{n},T_{n}){=}(j,t,\Delta,
$\linebreak$
\beta,Y)$ and for all $i{\in}(n)$, $(l_{i{-}1},v_{i{-}1},\Theta_{i{-}1},\delta_{i{-}1},T_{i{-}1}){\ll}(l_{i},v_{i},\Theta_{i},\delta_{i},T_{i})$.
By induction on $n$, the reader may easily verify that there exists $(m,w,\Phi,\epsilon,U){\in}{\mathcal T}$ such that $(m,w,\Phi,\epsilon,U){\ll^{\star}}(i,s,\Gamma,\alpha,X)$ and $(m,w,\Phi,\epsilon,U){\triangleright}(k,u,\Lambda,\gamma,Z)$.
\medskip
\end{proof}
\section{Defects}\label{section:defects}
\subsection*{Defects of $\rightarrow$-maximality}
A {\em defect of $\rightarrow$-maximality of a coherent and regular clip $({\mathcal T},{\ll},{\triangleright})$}\/ is a triple $((i,s,\Gamma,\alpha,X),B,C)$ where $(i,s,\Gamma,\alpha,X){\in}{\mathcal T}$ and $B,C{\in}\Fo$ are such that
\begin{itemize}
\item $B{\rightarrow}C{\in}\Gamma$,
\item $s{\not\models}B{\rightarrow}C$ and $s$ is not maximal with respect to $B{\rightarrow}C$,
\item for all $(j,t,\Delta,\beta,Y){\in}{\mathcal T}$, if $(i,s,\Gamma,\alpha,X){\ll}(j,t,\Delta,\beta,Y)$ then $t{\models}B{\rightarrow}C$.
\end{itemize}
The {\em rank of the defect $((i,s,\Gamma,\alpha,X),B,C)$ of $\rightarrow$-maximality}\/ is $\alpha$.
\\
\\
The {\em height of the defect $((i,s,\Gamma,\alpha,X),B,C)$ of $\rightarrow$-maximality}\/ is $X$.
\subsection*{Defects of $\square$-maximality}
A {\em defect of $\square$-maximality of a coherent and regular clip $({\mathcal T},{\ll},{\triangleright})$}\/ is a couple $((i,s,\Gamma,\alpha,X),B)$ where $(i,s,\Gamma,\alpha,X){\in}{\mathcal T}$ and $B{\in}\Fo$ are such that
\begin{itemize}
\item ${\square}B{\in}\Gamma$,
\item $s{\not\models}{\square}B$ and $s$ is not maximal with respect to ${\square}B$,
\item for all $(j,t,\Delta,\beta,Y){\in}{\mathcal T}$, if $(i,s,\Gamma,\alpha,X){\ll}(j,t,\Delta,\beta,Y)$ then $t{\models}{\square}B$.
\end{itemize}
The {\em rank of the defect $((i,s,\Gamma,\alpha,X),B)$ of $\square$-maximality}\/ is $\alpha$.
\\
\\
The {\em height of the defect $((i,s,\Gamma,\alpha,X),B)$ of $\square$-maximality}\/ is $X$.
\subsection*{Defects of $\square$-accessibility}
A {\em defect of $\square$-accessibility of a coherent and regular clip $({\mathcal T},{\ll},{\triangleright})$}\/ is a couple $((i,s,\Gamma,\alpha,X),B)$ where $(i,s,\Gamma,\alpha,X){\in}{\mathcal T}$ and $B{\in}\Fo$ are such that
\begin{itemize}
\item ${\square}B{\in}\Gamma$,
\item $s{\not\models}{\square}B$ and $s$ is maximal with respect to $\square B$,
\item for all $(j,t,\Delta,\beta,Y){\in}{\mathcal T}$, if $(i,s,\Gamma,\alpha,X){\triangleright}(j,t,\Delta,\beta,Y)$ then $t{\models}B$.
\end{itemize}
The {\em rank of the defect $((i,s,\Gamma,\alpha,X),B)$ of $\square$-accessibility}\/ is $\alpha$.
\\
\\
The {\em height of the defect $((i,s,\Gamma,\alpha,X),B)$ of $\square$-accessibility}\/ is $X$.
\subsection*{Defects of $\lozenge$-accessibility}
A {\em defect of $\lozenge$-accessibility of a coherent and regular clip $({\mathcal T},{\ll},{\triangleright})$}\/ is a couple $((i,s,\Gamma,\alpha,X),B)$ where $(i,s,\Gamma,\alpha,X){\in}{\mathcal T}$ and $B{\in}\Fo$ are such that
\begin{itemize}
\item ${\lozenge}B{\in}\Gamma$,
\item $s{\models}{\lozenge}B$,
\item for all $(j,t,\Delta,\beta,Y){\in}{\mathcal T}$, if $(i,s,\Gamma,\alpha,X){\triangleright}(j,t,\Delta,\beta,Y)$ then $t{\not\models}B$.
\end{itemize}
The {\em rank of the defect $((i,s,\Gamma,\alpha,X),B)$ of $\lozenge$-accessibility}\/ is $\alpha$.
\\
\\
The {\em height of the defect $((i,s,\Gamma,\alpha,X),B)$ of $\lozenge$-accessibility}\/ is $X$.
\subsection*{Defects of downward confluence}
When $\L{=}\LIK$, a {\em defect of downward confluence of a coherent and regular clip $({\mathcal T},{\ll},{\triangleright})$}\/ is a triple $((i,s,\Gamma,\alpha,X),(j,t,\Delta,\beta,Y),(k,u,\Lambda,\gamma,Z))$ where $(i,s,\Gamma,\alpha,
$\linebreak$
X),(j,t,\Delta,\beta,Y),(k,u,\Lambda,\gamma,Z){\in}{\mathcal T}$ are such that
\begin{itemize}
\item $(i,s,\Gamma,\alpha,X){\ll}(j,t,\Delta,\beta,Y)$,
\item $(j,t,\Delta,\beta,Y){\triangleright}(k,u,\Lambda,\gamma,Z)$,
\item for all $(l,v,\Theta,\delta,T){\in}{\mathcal T}$, either $(i,s,\Gamma,\alpha,X){\not\triangleright}(l,v,\Theta,\delta,T)$, or $(l,v,\Theta,\delta,T)
$\linebreak$
{\not\ll}(k,u,\Lambda,\gamma,Z)$.
\end{itemize}
The {\em rank of the defect $((i,s,\Gamma,\alpha,X),(j,t,\Delta,\beta,Y),(k,u,\Lambda,\gamma,Z))$ of downward confluence}\/ is $\alpha$.
\\
\\
The {\em height of the defect $((i,s,\Gamma,\alpha,X),(j,t,\Delta,\beta,Y),(k,u,\Lambda,\gamma,Z))$ of downward confluence}\/ is $X$.
\subsection*{Defects of forward confluence}
A {\em defect of forward confluence of a coherent and regular clip $({\mathcal T},{\ll},{\triangleright})$}\/ is a triple $((i,s,\Gamma,\alpha,X),(j,t,\Delta,\beta,Y),(k,u,\Lambda,\gamma,Z))$ where $(i,s,\Gamma,\alpha,X),(j,t,\Delta,\beta,Y),
$\linebreak$
(k,u,\Lambda,\gamma,Z){\in}{\mathcal T}$ are such that
\begin{itemize}
\item $(j,t,\Delta,\beta,Y){\ll}(i,s,\Gamma,\alpha,X)$,
\item $(j,t,\Delta,\beta,Y){\triangleright}(k,u,\Lambda,\gamma,Z)$,
\item for all $(l,v,\Theta,\delta,T){\in}{\mathcal T}$, either $(i,s,\Gamma,\alpha,X){\not\triangleright}(l,v,\Theta,\delta,T)$, or $(k,u,\Lambda,\gamma,Z)
$\linebreak$
{\not\ll}(l,v,\Theta,\delta,T)$.
\end{itemize}
The {\em rank of the defect $((i,s,\Gamma,\alpha,X),(j,t,\Delta,\beta,Y),(k,u,\Lambda,\gamma,Z))$ of forward confluence}\/ is $\alpha$.
\\
\\
The {\em height of the defect $((i,s,\Gamma,\alpha,X),(j,t,\Delta,\beta,Y),(k,u,\Lambda,\gamma,Z))$ of forward confluence}\/ is $X$.
\subsection*{Cleanness}
For all $\alpha{\in}\N$, we say that a coherent and regular clip $({\mathcal T},{\ll},{\triangleright})$ is
\begin{itemize}
\item {\em $\alpha$-clean for maximality}\/ if for all $\beta{\in}\N$, if $\beta{<}\alpha$ then $({\mathcal T},{\ll},{\triangleright})$ contains no defect of maximality of rank $\beta$,
\item {\em $\alpha$-clean for accessibility}\/ if for all $\beta{\in}\N$, if $\beta{<}\alpha$ then $({\mathcal T},{\ll},{\triangleright})$ contains no defect of accessibility of rank $\beta$,
\item {\em $\alpha$-clean for downward confluence}\/ if for all $\beta{\in}\N$, if $\beta{<}\alpha$ then $({\mathcal T},{\ll},{\triangleright})$ contains no defect of downward confluence of rank $\beta$,
\item {\em $\alpha$-clean for forward confluence}\/ if for all $\beta{\in}\N$, if $\beta{<}\alpha$ then $({\mathcal T},{\ll},{\triangleright})$ contains no defect of forward confluence of rank $\beta$.
\end{itemize}
We say that a coherent and regular clip $({\mathcal T},{\ll},{\triangleright})$ is {\em clean}\/ if it contains no defect.
\section{Repair of maximality defects}\label{section:reparations:maximality:defects}
\subsection*{Repair of $\rightarrow$-maximality defects}
The {\em repair of a defect $((i,s,\Gamma,\alpha,X),B,C)$ of $\rightarrow$-maximality of a coherent and regular clip $({\mathcal T},{\ll},{\triangleright})$}\/ consists in sequentially executing the following actions:
\begin{itemize}
\item add a tip $(j,t,\Delta,\beta,Y)$ to ${\mathcal T}$ such that $j$ is new, $s{<_{\L}}t$, $t{\not\models}B{\rightarrow}C$, $t$ is maximal with respect to $B{\rightarrow}C$, $\Delta{=}\Gamma$, $\beta{=}\alpha$ and $Y{=}X{+}1$,
\item add the couple $((i,s,\Gamma,\alpha,X),(j,t,\Delta,\beta,Y))$ to $\ll$.
\end{itemize}
Obviously, the resulting clip is coherent.
\\
\\
Moreover, since the resulting clip is obtained by adding the tip $(j,t,\Delta,\beta,Y)$ to ${\mathcal T}$ and the couple $((i,s,\Gamma,\alpha,X),(j,t,\Delta,\beta,Y))$ to $\ll$, then the resulting clip is regular.
Indeed, suppose for a while that the resulting clip is not regular.
Hence, the tip $(j,t,\Delta,\beta,Y)$ has a $\ll$-predecessor that is different from $(i,s,\Gamma,\alpha,X)$ in the resulting clip.
Since $j$ is new, then this is impossible.
\\
\\
In other respect, notice that if $\Gamma{=}\Sigma_{A}^{\alpha}$ then $\Delta{=}\Sigma_{A}^{\beta}$.
\\
\\
As well, notice that this repair is the repair of a defect of rank $\alpha$ and height $X$ that only introduces in ${\mathcal T}$ a tip of rank $\alpha$ and height $X{+}1$.
\\
\\
Finally, notice that $\degre(j,t,\Delta,\beta,Y){<}\degre(i,s,\Gamma,\alpha,X)$.
\subsection*{Repair of $\square$-maximality defects}
The {\em repair of a defect $((i,s,\Gamma,\alpha,X),B)$ of $\square$-maximality of a coherent and regular clip $({\mathcal T},{\ll},{\triangleright})$}\/ consists in sequentially executing the following actions:
\begin{itemize}
\item add a tip $(j,t,\Delta,\beta,Y)$ to ${\mathcal T}$ such that $j$ is new, $s{<_{\L}}t$, $t{\not\models}{\square}B$, $t$ is maximal with respect to ${\square}B$, $\Delta{=}\Gamma$, $\beta{=}\alpha$ and $Y{=}X{+}1$,
\item add the couple $((i,s,\Gamma,\alpha,X),(j,t,\Delta,\beta,Y))$ to $\ll$.
\end{itemize}
Obviously, the resulting clip is coherent.
\\
\\
Moreover, since the resulting clip is obtained by adding the tip $(j,t,\Delta,\beta,Y)$ to ${\mathcal T}$ and the couple $((i,s,\Gamma,\alpha,X),(j,t,\Delta,\beta,Y))$ to $\ll$, then the resulting clip is regular.
Indeed, suppose for a while that the resulting clip is not regular.
Hence, the tip $(j,t,\Delta,\beta,Y)$ has a $\ll$-predecessor that is different from $(i,s,\Gamma,\alpha,X)$ in the resulting clip.
Since $j$ is new, then this is impossible.
\\
\\
In other respect, notice that if $\Gamma{=}\Sigma_{A}^{\alpha}$ then $\Delta{=}\Sigma_{A}^{\beta}$.
\\
\\
As well, notice that this repair is the repair of a defect of rank $\alpha$ and height $X$ that only introduces in ${\mathcal T}$ a tip of rank $\alpha$ and height $X{+}1$.
\\
\\
Finally, notice that $\degre(j,t,\Delta,\beta,Y){<}\degre(i,s,\Gamma,\alpha,X)$.
\subsection*{The maximality procedure}
Given $\alpha{\in}\N$ and a coherent and regular clip $({\mathcal T},{\ll},{\triangleright})$, the {\em maximality procedure}\/ is defined as follows:
\begin{enumerate}
\item $x:=({\mathcal T},{\ll},{\triangleright})$,
\item $X:=0$,
\item while $x$ contains defects of maximality of rank $\alpha$ do
\begin{enumerate}
\item repair in $x$ all defects of maximality of rank $\alpha$ and height $X$,
\item $X:=X{+}1$.
\end{enumerate}
\end{enumerate}
Given $\alpha{\in}\N$ and a coherent and regular clip $({\mathcal T},{\ll},{\triangleright})$, the role of the maximality procedure is to iteratively repair all defects of maximality of $({\mathcal T},{\ll},{\triangleright})$ of rank $\alpha$.
\begin{lemma}\label{lemma:about:height:maximal:in:tips:repairing:max:defects}
Given $\alpha{\in}\N$ and a coherent and regular clip $({\mathcal T},{\ll},{\triangleright})$, at any moment of the execution of the maximality procedure, for all tips $(j,t,\Delta,\beta,Y)$ occurring in $x$, $Y{\leq}\max\{Z{\in}\N$: $(k,u,\Lambda,\gamma,Z){\in}{\mathcal T}\}{+}\card(\Sigma_{A})$.
\end{lemma}
\begin{proof}
It suffices to notice that for all $X{\in}\N$, the repair of a defect of maximality of rank $\alpha$ and height $X$ only introduces a tip of rank $\alpha$ and height $X{+}1$.
Moreover, as noticed above, the degree of the introduced tip is strictly smaller than the degree of the tip that has caused this introduction.
\medskip
\end{proof}
\begin{lemma}\label{lemma:procedure:defects:maximality:terminates}
Given $\alpha{\in}\N$ and a coherent and regular clip $({\mathcal T},{\ll},{\triangleright})$, the maximality procedure terminates.
\end{lemma}
\begin{proof}
By Lemma~\ref{lemma:about:height:maximal:in:tips:repairing:max:defects}.
\medskip
\end{proof}
\begin{lemma}\label{lemma:about:alpha:clean:coherent:clips:repairing:maximality}
Let $\alpha{\in}\N$ and $({\mathcal T},{\ll},{\triangleright})$ be a coherent and regular clip.
If $({\mathcal T},{\ll},{\triangleright})$ is $\alpha$-clean for maximality, $\alpha$-clean for accessibility, $\alpha$-clean for downward confluence when $\L{=}\LIK$ and $\alpha$-clean for forward confluence then the clip obtained from $({\mathcal T},{\ll},{\triangleright})$ after the execution of the maximality procedure is $(\alpha{+}1)$-clean for maximality, $\alpha$-clean for accessibility, $\alpha$-clean for downward confluence when $\L{=}\LIK$ and $\alpha$-clean for forward confluence.
\end{lemma}
\begin{proof}
Suppose $({\mathcal T},{\ll},{\triangleright})$ is $\alpha$-clean for maximality, $\alpha$-clean for accessibility, $\alpha$-clean for downward confluence when $\L{=}\LIK$ and $\alpha$-clean for forward confluence.
Notice that by Lemma~\ref{lemma:procedure:defects:maximality:terminates}, the maximality procedure terminates.
In other respect, since $({\mathcal T},{\ll},{\triangleright})$ is $\alpha$-clean for maximality, $\alpha$-clean for accessibility, $\alpha$-clean for downward confluence when $\L{=}\LIK$ and $\alpha$-clean for forward confluence and the execution of the maximality procedure only introduces tips of rank $\alpha$, then the clip obtained from $({\mathcal T},{\ll},{\triangleright})$ after the execution of the maximality procedure is $(\alpha{+}1)$-clean for maximality, $\alpha$-clean for accessibility, $\alpha$-clean for downward confluence when $\L{=}\LIK$ and $\alpha$-clean for forward confluence.
\medskip
\end{proof}
\section{Repair of accessibility defects}\label{section:reparations:accessibility:defects}
\subsection*{Repair of $\square$-accessibility defects}
The {\em repair of a defect $((i,s,\Gamma,\alpha,X),B)$ of $\square$-accessibility of a coherent and regular clip $({\mathcal T},{\ll},{\triangleright})$}\/ consists in sequentially executing the following actions:
\begin{itemize}
\item add a tip $(j,t,\Delta,\beta,Y)$ to ${\mathcal T}$ such that $j$ is new, $t{\in}R_{\L}(s)$, $t{\not\models}B$, $\Delta{=}\Gamma^{\circ}$, $\beta{=}\alpha{+}1$ and $Y{=}X$,
\item add the couple $((i,s,\Gamma,\alpha,X),(j,t,\Delta,\beta,Y))$ to $\triangleright$.
\end{itemize}
Obviously, the resulting clip is coherent.
\\
\\
Moreover, since the resulting clip is obtained by adding the tip $(j,t,\Delta,\beta,Y)$ to ${\mathcal T}$ and the couple $((i,s,\Gamma,\alpha,X),(j,t,\Delta,\beta,Y))$ to $\triangleright$, then the resulting clip is regular.
Indeed, suppose for a while that the resulting clip is not regular.
Hence, the tip $(j,t,\Delta,\beta,Y)$ has a $\triangleright$-predecessor that is different from $(i,s,\Gamma,\alpha,X)$ in the resulting clip.
Since $j$ is new, then this is impossible.
\\
\\
In other respect, notice that if $\Gamma{=}\Sigma_{A}^{\alpha}$ then $\Delta{=}\Sigma_{A}^{\beta}$.
\\
\\
Finally, notice that this repair is the repair of a defect of rank $\alpha$ and height $X$ that only introduces in ${\mathcal T}$ a tip of rank $\alpha{+}1$ and height $X$.
\subsection*{Repair of $\lozenge$-accessibility defects}
The {\em repair of a defect $((i,s,\Gamma,\alpha,X),B)$ of $\lozenge$-accessibility of a coherent and regular clip $({\mathcal T},{\ll},{\triangleright})$}\/ consists in sequentially executing the following actions:
\begin{itemize}
\item add a tip $(j,t,\Delta,\beta,Y)$ to ${\mathcal T}$ such that $j$ is new, $t{\in}R_{\L}(s)$, $t{\models}B$, $\Delta{=}\Gamma^{\circ}$, $\beta{=}\alpha{+}1$ and $Y{=}X$,
\item add the couple $((i,s,\Gamma,\alpha,X),(j,t,\Delta,\beta,Y))$ to $\triangleright$.
\end{itemize}
Obviously, the resulting clip is coherent.
\\
\\
Moreover, since the resulting clip is obtained by adding the tip $(j,t,\Delta,\beta,Y)$ to ${\mathcal T}$ and the couple $((i,s,\Gamma,\alpha,X),(j,t,\Delta,\beta,Y))$ to $\triangleright$, then the resulting clip is regular.
Indeed, suppose for a while that the resulting clip is not regular.
Hence, the tip $(j,t,\Delta,\beta,Y)$ has a $\triangleright$-predecessor that is different from $(i,s,\Gamma,\alpha,X)$ in the resulting clip.
Since $j$ is new, then this is impossible.
\\
\\
In other respect, notice that if $\Gamma{=}\Sigma_{A}^{\alpha}$ then $\Delta{=}\Sigma_{A}^{\beta}$.
\\
\\
Finally, notice that this repair is the repair of a defect of rank $\alpha$ and height $X$ that only introduces in ${\mathcal T}$ a tip of rank $\alpha{+}1$ and height $X$.
\subsection*{The accessibility procedure}
Given $\alpha{\in}\N$ and a coherent and regular clip $({\mathcal T},{\ll},{\triangleright})$, the {\em accessibility procedure}\/ is defined as follows:
\begin{enumerate}
\item $x:=({\mathcal T},{\ll},{\triangleright})$,
\item $X:=0$,
\item while $x$ contains defects of accessibility of rank $\alpha$ do
\begin{enumerate}
\item repair in $x$ all defects of accessibility of rank $\alpha$ and height $X$,
\item $X:=X{+}1$.
\end{enumerate}
\end{enumerate}
Given $\alpha{\in}\N$ and a coherent and regular clip $({\mathcal T},{\ll},{\triangleright})$, the role of the accessibility procedure is to iteratively repair all defects of accessibility of $({\mathcal T},{\ll},{\triangleright})$ of rank $\alpha$.
\begin{lemma}\label{lemma:about:height:maximal:in:tips:repairing:accessibility:defects}
Given $\alpha{\in}\N$ and a coherent and regular clip $({\mathcal T},{\ll},{\triangleright})$, at any moment of the execution of the accessibility procedure, for all tips $(j,t,\Delta,\beta,Y)$ occurring in $x$, $Y{\leq}\max\{Z{\in}\N$: $(k,u,\Lambda,\gamma,Z){\in}{\mathcal T}\}$.
\end{lemma}
\begin{proof}
It suffices to notice that for all $X{\in}\N$, the repair of a defect of accessibility of rank $\alpha$ and height $X$ only introduces a tip of rank $\alpha{+}1$ and height $X$.
\medskip
\end{proof}
\begin{lemma}\label{lemma:procedure:defects:accessibility:terminates}
Given $\alpha{\in}\N$ and a coherent and regular clip $({\mathcal T},{\ll},{\triangleright})$, the accessibility procedure terminates.
\end{lemma}
\begin{proof}
By Lemma~\ref{lemma:about:height:maximal:in:tips:repairing:accessibility:defects}.
\medskip
\end{proof}
\begin{lemma}\label{lemma:about:alpha:clean:coherent:clips:repairing:accessibility}
Let $\alpha{\in}\N$ and $({\mathcal T},{\ll},{\triangleright})$ be a coherent and regular clip.
If $({\mathcal T},{\ll},
$\linebreak$
{\triangleright})$ is $(\alpha{+}1)$-clean for maximality, $\alpha$-clean for accessibility, $\alpha$-clean for downward confluence when $\L{=}\LIK$ and $\alpha$-clean for forward confluence then the clip obtained from $({\mathcal T},{\ll},{\triangleright})$ after the execution of the accessibility procedure is $(\alpha{+}1)$-clean for maximality, $(\alpha{+}1)$-clean for accessibility, $\alpha$-clean for downward confluence when $\L{=}\LIK$ and $\alpha$-clean for forward confluence.
\end{lemma}
\begin{proof}
Suppose $({\mathcal T},{\ll},{\triangleright})$ is $(\alpha{+}1)$-clean for maximality, $\alpha$-clean for accessibility, $\alpha$-clean for downward confluence when $\L{=}\LIK$ and $\alpha$-clean for forward confluence.
Notice that by Lemma~\ref{lemma:procedure:defects:accessibility:terminates}, the accessibility procedure terminates.
In other respect, since $({\mathcal T},{\ll},{\triangleright})$ is $(\alpha{+}1)$-clean for maximality, $\alpha$-clean for accessibility, $\alpha$-clean for downward confluence when $\L{=}\LIK$ and $\alpha$-clean for forward confluence and the execution of the accessibility procedure only introduces tips of rank $\alpha{+}1$, then the clip obtained from $({\mathcal T},{\ll},{\triangleright})$ after the execution of the accessibility procedure is $(\alpha{+}1)$-clean for maximality, $(\alpha{+}1)$-clean for accessibility, $\alpha$-clean for downward confluence when $\L{=}\LIK$ and $\alpha$-clean for forward confluence.
\medskip
\end{proof}
\section{Repair of confluence defects}\label{section:reparations:confluence:defects}
\subsection*{Repair of downward confluence defects}
When $\L{=}\LIK$, the {\em repair of a defect $((i,s,\Gamma,\alpha,X),(j,t,\Delta,\beta,Y),(k,u,\Lambda,\gamma,Z))$ of downward confluence of a coherent and regular clip $({\mathcal T},{\ll},{\triangleright})$}\/ consists in sequentially executing the following actions:
\begin{itemize}
\item add a tip $(l,v,\Theta,\delta,T)$ to ${\mathcal T}$ such that $l$ is new, $v{\in}R_{\L}(s)$, $v{\leq_{\L}}u$, $\Theta{=}\Gamma^{\circ}$, $\Theta{=}\Lambda$, $\delta{=}\alpha{+}1$, $\delta{=}\gamma$, $T{=}X$ and $T{=}Z{-}1$,
\item add the couple $((i,s,\Gamma,\alpha,X),(l,v,\Theta,\delta,T))$ to $\triangleright$,
\item add the couple $((l,v,\Theta,\delta,T),(k,u,\Lambda,\gamma,Z))$ to $\ll$.
\end{itemize}
Obviously, the resulting clip is coherent.
\\
\\
Moreover, since the resulting clip is obtained by adding the tip $(l,v,\Theta,\delta,T)$ to ${\mathcal T}$, the couple $((i,s,\Gamma,\alpha,X),(l,v,\Theta,\delta,T))$ to $\triangleright$ and the couple $((l,v,\Theta,\delta,T),(k,u,
$\linebreak$
\Lambda,\gamma,Z))$ to $\ll$, then the resulting clip is regular.
Indeed, suppose for a while that the resulting clip is not regular.
Hence, either the tip $(l,v,\Theta,\delta,T)$ has a $\triangleright$-predecessor that is different from $(i,s,\Gamma,\alpha,X)$ in the resulting clip, or the tip $(k,u,\Lambda,\gamma,Z)$ has a $\ll$-predecessor that is different from $(l,v,\Theta,\delta,T)$ in the resulting clip, or the tip $(l,v,\Theta,\delta,T)$ has a $\ll$-predecessor $(m,w,\Phi,\epsilon,U)$ such that no $\ll$-predecessor of $(i,s,\Gamma,\alpha,X)$ is a $\triangleright$-predecessor of $(m,w,\Phi,\epsilon,U)$ in the resulting clip, or the tip $(k,u,\Lambda,\gamma,Z)$ has a $\triangleright$-predecessor $(m,w,\Phi,\epsilon,U)$ such that no $\ll$-predecessor of $(m,w,\Phi,\epsilon,U)$ is a $\triangleright$-predecessor of $(l,v,\Theta,\delta,T)$ in the resulting clip.
In the first case, since $l$ is new, then this is impossible.
In the second case, consequently, there exists a $\ll$-predecessor $(m,w,\Phi,\epsilon,U)$ of $(k,u,\Lambda,\gamma,Z)$ in ${\mathcal T}$.
Since $(j,t,\Delta,\beta,Y)$ is a $\triangleright$-predecessor of $(k,u,\Lambda,\gamma,Z)$ in ${\mathcal T}$, $(i,s,\Gamma,\alpha,X)$ is a $\ll$-predecessor of $(j,t,\Delta,\beta,Y)$ in ${\mathcal T}$ and ${\mathcal T}$ is regular, then 
$(i,s,\Gamma,\alpha,X)$ is a $\triangleright$-predecessor of $(m,w,\Phi,\epsilon,U)$ in ${\mathcal T}$.
Since $((i,s,\Gamma,\alpha,X),(j,t,\Delta,\beta,Y),(k,u,\Lambda,\gamma,
$\linebreak$
Z))$ is a defect of downward confluence of $({\mathcal T},{\ll},{\triangleright})$, then this is impossible.
In the third case, since $l$ is new, then this is impossible.
In the fourth case, since $(j,t,\Delta,\beta,Y)$ is a $\triangleright$-predecessor of $(k,u,\Lambda,\gamma,Z)$ in ${\mathcal T}$, $(i,s,\Gamma,\alpha,X)$ is a $\ll$-predecessor of $(j,t,\Delta,\beta,Y)$ in ${\mathcal T}$ and ${\mathcal T}$ is regular, then 
$(i,s,\Gamma,\alpha,X)$ is a $\ll$-predecessor of $(m,w,\Phi,\epsilon,U)$ in ${\mathcal T}$ and a $\triangleright$-predecessor of $(l,v,\Theta,\delta,T)$ in the resulting clip: a contradiction.
\\
\\
In other respect, notice that if $\Gamma{=}\Sigma_{A}^{\alpha}$, $\Delta{=}\Sigma_{A}^{\beta}$ and $\Lambda{=}\Sigma_{A}^{\gamma}$ then $\Theta{=}\Sigma_{A}^{\delta}$.
\\
\\
Finally, notice that this repair is the repair of a defect of rank $\alpha$ and height $X$ that only introduces in ${\mathcal T}$ a tip of rank $\alpha{+}1$ and height $X$.
\subsection*{The procedure of downward confluence}
When $\L{=}\LIK$, given $\alpha{\in}\N$ and a coherent and regular clip $({\mathcal T},{\ll},{\triangleright})$, the {\em procedure of downward confluence}\/ is defined as follows:
\begin{enumerate}
\item $x:=({\mathcal T},{\ll},{\triangleright})$,
\item $X:=\hauteur_{\alpha}({\mathcal T})$,
\item while $x$ contains defects of downward confluence of rank $\alpha$ do
\begin{enumerate}
\item repair in $x$ all defects of downward confluence of rank $\alpha$ and height $X$,
\item $X:=X{-}1$.
\end{enumerate}
\end{enumerate}
When $\L{=}\LIK$, given $\alpha{\in}\N$ and a coherent and regular clip $({\mathcal T},{\ll},{\triangleright})$, the role of the procedure of downward confluence is to iteratively repair all defects of downward confluence of $({\mathcal T},{\ll},{\triangleright})$ of rank $\alpha$.
\begin{lemma}\label{lemma:about:height:maximal:in:tips:repairing:downward:confluence:defects}
When $\L{=}\LIK$, given $\alpha{\in}\N$ and a coherent and regular clip $({\mathcal T},
$\linebreak$
{\ll},{\triangleright})$, at any moment of the execution of the procedure of downward confluence, for all tips $(j,t,\Delta,\beta,Y)$ occurring in $x$, $Y{\leq}\max\{Z{\in}\N$: $(k,u,\Lambda,\gamma,Z){\in}{\mathcal T}\}$.
\end{lemma}
\begin{proof}
It suffices to notice that for all $X{\in}\N$, the repair of a defect of downward confluence of rank $\alpha$ and height $X$ only introduces a tip of rank $\alpha{+}1$ and height $X$.
\medskip
\end{proof}
\begin{lemma}\label{lemma:procedure:defects:downward:confluence:terminates}
When $\L{=}\LIK$, given $\alpha{\in}\N$ and a coherent and regular clip $({\mathcal T},
$\linebreak$
{\ll},{\triangleright})$, the procedure of downward confluence terminates.
\end{lemma}
\begin{proof}
By Lemma~\ref{lemma:about:height:maximal:in:tips:repairing:downward:confluence:defects}.
\medskip
\end{proof}
\begin{lemma}\label{lemma:about:alpha:clean:coherent:clips:repairing:downward:confluence}
Let $\alpha{\in}\N$ and $({\mathcal T},{\ll},{\triangleright})$ be a coherent and regular clip.
When $\L{=}\LIK$, if $({\mathcal T},{\ll},{\triangleright})$ is $(\alpha{+}1)$-clean for maximality, $(\alpha{+}1)$-clean for accessibility, $\alpha$-clean for downward confluence and $\alpha$-clean for forward confluence then the clip obtained from $({\mathcal T},{\ll},{\triangleright})$ after the execution of the procedure of downward confluence is $(\alpha{+}1)$-clean for maximality, $(\alpha{+}1)$-clean for accessibility, $(\alpha{+}1)$-clean for downward confluence and $\alpha$-clean for forward confluence.
\end{lemma}
\begin{proof}
Suppose $({\mathcal T},{\ll},{\triangleright})$ is $(\alpha{+}1)$-clean for maximality, $(\alpha{+}1)$-clean for accessibility, $\alpha$-clean for downward confluence and $\alpha$-clean for forward confluence.
Notice that by Lemma~\ref{lemma:procedure:defects:downward:confluence:terminates}, the procedure of downward confluence terminates.
In other respect, since $({\mathcal T},{\ll},{\triangleright})$ is $(\alpha{+}1)$-clean for maximality, $(\alpha{+}1)$-clean for accessibility, $\alpha$-clean for downward confluence and $\alpha$-clean for forward confluence and the execution of the procedure of downward confluence only introduces tips of rank $\alpha{+}1$, then the clip obtained from $({\mathcal T},{\ll},{\triangleright})$ after the execution of the procedure of downward confluence is $(\alpha{+}1)$-clean for maximality, $(\alpha{+}1)$-clean for accessibility, $(\alpha{+}1)$-clean for downward confluence and $\alpha$-clean for forward confluence.
\medskip
\end{proof}
\subsection*{Repair of forward confluence defects}
The {\em repair of a defect $((i,s,\Gamma,\alpha,X),(j,t,\Delta,\beta,Y),(k,u,\Lambda,\gamma,Z))$ of forward confluence of a coherent and regular clip $({\mathcal T},{\ll},{\triangleright})$}\/ consists in sequentially executing the following actions:
\begin{itemize}
\item add a tip $(l,v,\Theta,\delta,T)$ to ${\mathcal T}$ such that $l$ is new, $v{\in}R_{\L}(s)$, $u{\leq_{\L}}v$, $\Theta{=}\Gamma^{\circ}$, $\Theta{=}\Lambda$, $\delta{=}\alpha{+}1$, $\delta{=}\gamma$, $T{=}X$ and $T{=}Z{+}1$,
\item add the couple $((i,s,\Gamma,\alpha,X),(l,v,\Theta,\delta,T))$ to $\triangleright$,
\item add the couple $((k,u,\Lambda,\gamma,Z),(l,v,\Theta,\delta,T))$ to $\ll$.
\end{itemize}
Obviously, the resulting clip is coherent.
\\
\\
Moreover, since the resulting clip is obtained by adding the tip $(l,v,\Theta,\delta,T)$ to ${\mathcal T}$, the couple $((i,s,\Gamma,\alpha,X),(l,v,\Theta,\delta,T))$ to $\triangleright$ and the couple $((k,u,\Lambda,\gamma,Z),(l,v,
$\linebreak$
\Theta,\delta,T))$ to $\ll$, then the resulting clip is regular.
Indeed, suppose for a while that the resulting clip is not regular.
Hence, either the tip $(l,v,\Theta,\delta,T)$ has a $\triangleright$-predecessor that is different from $(i,s,\Gamma,\alpha,X)$ in the resulting clip, or the tip $(l,v,\Theta,\delta,T)$ has a $\ll$-predecessor that is different from $(k,u,\Lambda,\gamma,Z)$ in the resulting clip, or the tip $(l,v,\Theta,\delta,T)$ has a $\ll$-predecessor $(m,w,\Phi,\epsilon,U)$ such that no $\ll$-predecessor of $(i,s,\Gamma,\alpha,X)$ is a $\triangleright$-predecessor of $(m,w,\Phi,\epsilon,U)$ in the resulting clip, or the tip $(l,v,\Theta,\delta,T)$ has a $\triangleright$-predecessor $(m,w,\Phi,\epsilon,U)$ such that no $\ll$-predecessor of $(m,w,\Phi,\epsilon,U)$ is a $\triangleright$-predecessor of $(k,u,\Lambda,\gamma,Z)$ in the resulting clip.
In the first case, since $l$ is new, then this is impossible.
In the second case, since $l$ is new, then this is impossible.
In the third case, since $l$ is new, then $(m,w,\Phi,\epsilon,U)$ is equal to $(k,u,\Lambda,\gamma,Z)$.
Since $(j,t,\Delta,\beta,Y)$ is a $\ll$-predecessor of $(i,s,\Gamma,\alpha,X)$ and no $\ll$-predecessor of $(i,s,\Gamma,\alpha,X)$ is a $\triangleright$-predecessor of $(m,w,\Phi,\epsilon,U)$ in the resulting clip, then $(j,t,\Delta,\beta,Y)$ is not a $\triangleright$-predecessor of $(k,u,\Lambda,\gamma,Z)$ in ${\mathcal T}$: a contradiction.
In the fourth case, since $l$ is new, then $(m,w,\Phi,\epsilon,U)$ is equal to $(i,s,\Gamma,\alpha,X)$.
Since $(j,t,\Delta,\beta,Y)$ is a $\ll$-predecessor of $(i,s,\Gamma,\alpha,X)$ and no $\ll$-predecessor of $(m,w,\Phi,\epsilon,U)$ is a $\triangleright$-predecessor of $(k,u,\Lambda,\gamma,Z)$ in the resulting clip, then $(j,t,\Delta,\beta,Y)$ is not a $\triangleright$-predecessor of $(k,u,\Lambda,\gamma,Z)$ in ${\mathcal T}$: a contradiction.
\\
\\
In other respect, notice that if $\Gamma{=}\Sigma_{A}^{\alpha}$, $\Delta{=}\Sigma_{A}^{\beta}$ and $\Lambda{=}\Sigma_{A}^{\gamma}$ then $\Theta{=}\Sigma_{A}^{\delta}$.
\\
\\
Finally, notice that this repair is the repair of a defect of rank $\alpha$ and height $X$ that only introduces in ${\mathcal T}$ a tip of rank $\alpha{+}1$ and height $X$.
\subsection*{The procedure of forward confluence}
Given $\alpha{\in}\N$ and a coherent and regular clip $({\mathcal T},{\ll},{\triangleright})$, the {\em procedure of forward confluence}\/ is defined as follows:
\begin{enumerate}
\item $x:=({\mathcal T},{\ll},{\triangleright})$,
\item $X:=0$,
\item while $x$ contains defects of forward confluence of rank $\alpha$ do
\begin{enumerate}
\item repair in $x$ all defects of forward confluence of rank $\alpha$ and height $X$,
\item $X:=X{+}1$.
\end{enumerate}
\end{enumerate}
Given $\alpha{\in}\N$ and a coherent and regular clip $({\mathcal T},{\ll},{\triangleright})$, the role of the procedure of forward confluence is to iteratively repair all defects of forward confluence of $({\mathcal T},{\ll},{\triangleright})$ of rank $\alpha$.
\begin{lemma}\label{lemma:about:height:maximal:in:tips:repairing:forward:confluence:defects}
Given $\alpha{\in}\N$ and a coherent and regular clip $({\mathcal T},{\ll},{\triangleright})$, at any moment of the execution of the procedure of forward confluence, for all tips $(j,t,\Delta,
$\linebreak$
\beta,Y)$ occurring in $x$, $Y{\leq}\max\{Z{\in}\N$: $(k,u,\Lambda,\gamma,Z){\in}{\mathcal T}\}$.
\end{lemma}
\begin{proof}
It suffices to notice that for all $X{\in}\N$, the repair of a defect of forward confluence of rank $\alpha$ and height $X$ only introduces a tip of rank $\alpha{+}1$ and height $X$.
\medskip
\end{proof}
\begin{lemma}\label{lemma:procedure:defects:forward:confluence:terminates}
Given $\alpha{\in}\N$ and a coherent and regular clip $({\mathcal T},{\ll},{\triangleright})$, the procedure of forward confluence terminates.
\end{lemma}
\begin{proof}
By Lemma~\ref{lemma:about:height:maximal:in:tips:repairing:forward:confluence:defects}.
\medskip
\end{proof}
\begin{lemma}\label{lemma:about:alpha:clean:coherent:clips:repairing:forward:confluence}
Let $\alpha{\in}\N$ and $({\mathcal T},{\ll},{\triangleright})$ be a coherent and regular clip.
If $({\mathcal T},{\ll},
$\linebreak$
{\triangleright})$ is $(\alpha{+}1)$-clean for maximality, $(\alpha{+}1)$-clean for accessibility, $(\alpha{+}1)$-clean for downward confluence when $\L{=}\LIK$ and $\alpha$-clean for forward confluence then the clip obtained from $({\mathcal T},{\ll},{\triangleright})$ after the execution of the procedure of forward confluence is $(\alpha{+}1)$-clean for maximality, $(\alpha{+}1)$-clean for accessibility, $(\alpha{+}1)$-clean for downward confluence when $\L{=}\LIK$ and $(\alpha{+}1)$-clean for forward confluence.
\end{lemma}
\begin{proof}
Suppose $({\mathcal T},{\ll},{\triangleright})$ is $(\alpha{+}1)$-clean for maximality, $(\alpha{+}1)$-clean for accessibility, $(\alpha{+}1)$-clean for downward confluence when $\L{=}\LIK$ and $\alpha$-clean for forward confluence.
Notice that by Lemma~\ref{lemma:procedure:defects:forward:confluence:terminates}, the procedure of forward confluence terminates.
In other respect, since $({\mathcal T},{\ll},{\triangleright})$ is $(\alpha{+}1)$-clean for maximality, $(\alpha{+}1)$-clean for accessibility, $(\alpha{+}1)$-clean for downward confluence when $\L{=}\LIK$ and $\alpha$-clean for forward confluence and the execution of the procedure of forward confluence only introduces tips of rank $\alpha{+}1$, then the clip obtained from $({\mathcal T},{\ll},{\triangleright})$ after the execution of the procedure of forward confluence is $(\alpha{+}1)$-clean for maximality, $(\alpha{+}1)$-clean for accessibility, $(\alpha{+}1)$-clean for downward confluence when $\L{=}\LIK$ and $(\alpha{+}1)$-clean for forward confluence.
\medskip
\end{proof}
\section{A decision procedure}\label{section:a:decision:procedure}
\subsection*{The saturation procedure}
The {\em saturation procedure}\/ is defined as follows:
\begin{enumerate}
\item $x:=(\{(0,s_{0},\Sigma_{A},0,0)\},\emptyset,\emptyset)$,
\item $\alpha:=0$,
\item while $x$ is not clean do
\begin{enumerate}
\item repair in $x$ all defects of maximality of rank $\alpha$ by applying the maximality procedure,
\item repair in $x$ all defects of accessibility of rank $\alpha$ by applying the accessibility procedure,
\item when $\L{=}\LIK$, repair in $x$ all defects of downward confluence of rank $\alpha$ by applying the procedure of downward confluence,
\item repair in $x$ all defects of forward confluence of rank $\alpha$ by applying the procedure of forward confluence,
\item $\alpha:=\alpha{+}1$.
\end{enumerate}
\end{enumerate}
The role of the saturation procedure is to iteratively repair all defects of $(\{(0,s_{0},
$\linebreak$
\Sigma_{A},0,0)\},\emptyset,\emptyset)$.
\begin{lemma}\label{lemma:invariant:of:the:saturation:procedure}
At any moment of the execution of the saturation procedure, for all tips $(j,t,\Delta,\beta,Y)$ occurring in $x$, $\beta{\leq}\card(\Sigma_{A}){-}\card(\Delta)$.
\end{lemma}
\begin{proof}
It suffices to notice that at any moment of the execution of the saturation procedure, for all tips $(j,t,\Delta,\beta,Y)$ occurring in $x$, $\Delta{=}\Sigma_{A}^{\beta}$.
\medskip
\end{proof}
\begin{proposition}\label{proposition:principal}
The saturation procedure terminates.
\end{proposition}
\begin{proof}
By Lemmas~\ref{lemma:procedure:defects:maximality:terminates}, \ref{lemma:about:alpha:clean:coherent:clips:repairing:maximality}, \ref{lemma:procedure:defects:accessibility:terminates}, \ref{lemma:about:alpha:clean:coherent:clips:repairing:accessibility}, \ref{lemma:procedure:defects:downward:confluence:terminates}, \ref{lemma:about:alpha:clean:coherent:clips:repairing:downward:confluence}, \ref{lemma:procedure:defects:forward:confluence:terminates}, \ref{lemma:about:alpha:clean:coherent:clips:repairing:forward:confluence} and~\ref{lemma:invariant:of:the:saturation:procedure}.
\medskip
\end{proof}
Therefore, there exists $\alpha{\in}\N$ and there exists a finite sequence $(({\mathcal T}_{0},{\ll_{0}},{\triangleright_{0}}),\ldots,
$\linebreak$
({\mathcal T}_{\alpha},{\ll_{\alpha}},{\triangleright_{\alpha}}))$ of clips such that
\begin{itemize}
\item $({\mathcal T}_{0},{\ll_{0}},{\triangleright_{0}}){=}(\{(0,s_{0},\Sigma_{A},0,0)\},\emptyset,\emptyset)$,
\item for all $\beta{\in}\N$, if $\beta{<}\alpha$ then $({\mathcal T}_{\beta},{\ll_{\beta}},{\triangleright_{\beta}})$ is not clean and $({\mathcal T}_{\beta{+}1},{\ll_{\beta{+}1}},{\triangleright_{\beta{+}1}})$ is obtained from $({\mathcal T}_{\beta},{\ll_{\beta}},{\triangleright_{\beta}})$ by firstly applying the maximality procedure iteratively repairing all defects of maximality of rank $\beta$, by secondly applying the accessibility procedure iteratively repairing all defects of accessibility of rank $\beta$, by thirdly applying the procedure of downward confluence iteratively repairing all defects of downward confluence of rank $\beta$ when $\L{=}\LIK$ and by fourthly applying the procedure of forward confluence iteratively repairing all defects of forward confluence of rank $\beta$,
\item $({\mathcal T}_{\alpha},{\ll_{\alpha}},{\triangleright_{\alpha}})$ is clean.
\end{itemize}
\section{Finite frame property}\label{section:finite:model:property}
\subsection*{The saturated frame}
Let $(W^{\prime},{\leq^{\prime}},{R^{\prime}})$ be the frame such that
\begin{itemize}
\item $W^{\prime}{=}{\mathcal T}_{\alpha}$,
\item ${\leq^{\prime}}{=}{\ll_{\alpha}^{\star}}$,
\item ${R^{\prime}}{=}{\triangleright_{\alpha}}$.
\end{itemize}
The frame $(W^{\prime},{\leq^{\prime}},{R^{\prime}})$ is called {\em saturated frame of $s_{0}$.}
\begin{lemma}\label{lemma:frame:prime:is:appropriate}
\begin{itemize}
\item $(W^{\prime},{\leq^{\prime}},{R^{\prime}})$ is in ${\mathcal C}_{\fucfra}^{\omega}$ when $\L{=}\FIK$,
\item $(W^{\prime},{\leq^{\prime}},{R^{\prime}})$ is in ${\mathcal C}_{\fducfra}^{\omega}$ when $\L{=}\LIK$.
\end{itemize}
\end{lemma}
\begin{proof}
By Lemma~\ref{lemma:about:regularity:and:upward:confluence}.
\medskip
\end{proof}
\subsection*{The saturated valuation}
The {\em saturated valuation of $s_{0}$}\/ is the valuation $V^{\prime}$: $\At\longrightarrow\wp(W^{\prime})$ on $(W^{\prime},{\leq^{\prime}},{R^{\prime}})$ such that for all atoms $p$, $V^{\prime}(p){=}\{(i,s,\Gamma,\alpha,X){\in}W^{\prime}$: $s{\in}V_{\L}(p)\}$.
\subsection*{The saturated model}
The {\em saturated model of $s_{0}$}\/ is the model $(W^{\prime},{\leq^{\prime}},{R^{\prime}},V^{\prime})$.
\subsection*{The Saturated Truth Lemma}
In Lemma~\ref{lemma:truth:lemma:at:the:end:of:the:procedure}, for all $B{\in}\Fo$ and for all $(i,s,\Gamma,\alpha,X){\in}W^{\prime}$, when we write ``$(i,s,\Gamma,
$\linebreak$
\alpha,X){\models}B$'' we mean ``$(W^{\prime},{\leq^{\prime}},{R^{\prime}},V^{\prime}),(i,s,\Gamma,\alpha,X){\models}B$'' and when we write ``$s{\models}
$\linebreak$
B$'' we mean ``$(W_{\L},{\leq_{\L}},{R_{\L}},V_{\L}),s{\models}B$''.
\begin{lemma}[Saturated Truth Lemma]\label{lemma:truth:lemma:at:the:end:of:the:procedure}
Let $B{\in}\Fo$.
For all $(i,s,\Gamma,\alpha,X){\in}
$\linebreak$
W^{\prime}$, if $B{\in}\Gamma$ then $(i,s,\Gamma,\alpha,X){\models}B$ if and only if $s{\models}B$.
\end{lemma}
\begin{proof}
By induction on $B$.
{\bf Case $B{=}p$:}
Let $(i,s,\Gamma,\alpha,X){\in}W^{\prime}$.
Suppose $p{\in}\Gamma$.
From left to right, suppose $(i,s,\Gamma,\alpha,X){\models}p$.
Hence, $(i,s,\Gamma,\alpha,X){\in}V^{\prime}(p)$.
Thus, $s{\in}V_{\L}(p)$.
Consequently, $s{\models}p$.
From right to left, suppose $s{\models}p$.
Hence, $s{\in}V_{\L}(p)$.
Thus, $(i,s,\Gamma,\alpha,X){\in}V^{\prime}(p)$.
Consequently, $(i,s,\Gamma,\alpha,X){\models}p$.
{\bf Case $B{=}C{\rightarrow}D$:}
Let $(i,s,\Gamma,\alpha,X){\in}W^{\prime}$.
Suppose $C{\rightarrow}D{\in}\Gamma$.
Hence, $C{\in}\Gamma$ and $D{\in}\Gamma$.
From left to right, suppose $(i,s,\Gamma,\alpha,X){\models}C{\rightarrow}D$.
For the sake of the contradiction, suppose $s{\not\models}C{\rightarrow}D$.
Since $({\mathcal T}_{\alpha},{\ll_{\alpha}},{\triangleright_{\alpha}})$ is clean, then there exists $(j,t,\Delta,\beta,Y){\in}W^{\prime}$ such that $(i,s,\Gamma,\alpha,X){\leq^{\prime}}(j,t,\Delta,\beta,Y)$, $t{\not\models}C{\rightarrow}D$ and $t$ is maximal with respect to $C{\rightarrow}D$.
Thus, by Lemma~\ref{lemma:maximal:rightarrow}, $t{\models}C$ and $t{\not\models}D$.
Since $C{\in}\Gamma$ and $D{\in}\Gamma$, then by induction hypothesis, $(j,t,\Delta,\beta,Y){\models}C$ and $(j,t,\Delta,\beta,Y){\not\models}D$.
Since $(i,s,\Gamma,\alpha,X){\leq^{\prime}}(j,t,\Delta,\beta,Y)$, then $(i,s,\Gamma,\alpha,X){\not\models}C{\rightarrow}D$: a contradiction.
From right to left, suppose $s{\models}C{\rightarrow}D$.
For the sake of the contradiction, suppose $(i,s,\Gamma,\alpha,X){\not\models}C{\rightarrow}D$.
Consequently, there exists $(j,t,\Delta,\beta,Y){\in}W^{\prime}$ such that $(i,s,\Gamma,\alpha,X){\leq^{\prime}}(j,t,\Delta,\beta,Y)$, $(j,t,\Delta,\beta,Y){\models}C$ and $(j,t,\Delta,\beta,Y){\not\models}D$.
Hence, $s{\leq_{\L}}t$.
Moreover, since $C{\in}\Gamma$ and $D{\in}\Gamma$, then by induction hypothesis, $t{\models}C$ and $t{\not\models}D$.
Thus, $s{\not\models}C{\rightarrow}D$: a contradiction.
{\bf Cases $B{=}\top$, $B{=}\bot$, $B{=}C{\vee}D$ and $B{=}C{\wedge}D$:}
These cases are left as exercises to the reader.
{\bf Case $B{=}{\square}C$:}
Let $(i,s,\Gamma,\alpha,X){\in}W^{\prime}$.
Suppose ${\square}C{\in}\Gamma$.
Consequently, $C{\in}\Gamma^{\circ}$.
From left to right, suppose $(i,s,\Gamma,\alpha,X){\models}{\square}C$.
For the sake of the contradiction, suppose $s{\not\models}{\square}C$.
Since $({\mathcal T}_{\alpha},{\ll_{\alpha}},{\triangleright_{\alpha}})$ is clean, then there exists $(j,t,\Delta,\beta,Y){\in}W^{\prime}$ such that $(i,s,\Gamma,\alpha,X){\leq^{\prime}}(j,t,\Delta,\beta,Y)$, $t{\not\models}{\square}C$ and $t$ is maximal with respect to ${\square}C$.
Since $({\mathcal T}_{\alpha},{\ll_{\alpha}},{\triangleright_{\alpha}})$ is clean, then there exists $(k,u,\Lambda,\gamma,Z){\in}W^{\prime}$ such that $(j,t,\Delta,\beta,Y){R^{\prime}}(k,u,\Lambda,\gamma,Z)$ and $u{\not\models}C$.
Since $C{\in}\Gamma^{\circ}$, then by induction hypothesis, $(k,u,\Lambda,\gamma,Z){\not\models}C$.
Since $(i,s,\Gamma,\alpha,X){\leq^{\prime}}(j,t,\Delta,\beta,Y)$ and $(j,t,\Delta,\beta,Y){R^{\prime}}
$\linebreak$
(k,u,\Lambda,\gamma,Z)$, then $(i,s,\Gamma,\alpha,X){\not\models}{\square}C$: a contradiction.
From right to left, suppose $s{\models}{\square}C$.
For the sake of the contradiction, suppose $(i,s,\Gamma,\alpha,X){\not\models}{\square}C$.
Hence, there exists $(j,t,\Delta,\beta,Y),(k,u,\Lambda,\gamma,Z){\in}W^{\prime}$ such that $(i,s,\Gamma,\alpha,X){\leq^{\prime}}(j,
$\linebreak$
t,\Delta,\beta,Y)$, $(j,t,\Delta,\beta,Y){R^{\prime}}(k,u,\Lambda,\gamma,Z)$ and $(k,u,\Lambda,\gamma,Z){\not\models}C$.
Thus, $s{\leq}t$ and $u{\in}R_{\L}(t)$.
Moreover, since $C{\in}\Gamma^{\circ}$, then by induction hypothesis, $u{\not\models}C$.
Consequently, $s{\not\models}{\square}C$: a contradiction.
{\bf Case $B{=}{\lozenge}C$:}
Let $(i,s,\Gamma,\alpha,X){\in}W^{\prime}$.
Suppose ${\lozenge}C{\in}\Gamma$.
Hence, $C{\in}\Gamma^{\circ}$.
From left to right, suppose $(i,s,\Gamma,\alpha,X){\models}{\lozenge}C$.
Thus, there exists $(j,t,\Delta,\beta,Y),(k,u,
$\linebreak$
\Lambda,\gamma,Z){\in}W^{\prime}$ such that $(j,t,\Delta,\beta,Y){\leq^{\prime}}(i,s,\Gamma,\alpha,X)$, $(j,t,\Delta,\beta,Y){R^{\prime}}(k,u,\Lambda,\gamma,
$\linebreak$
Z)$ and $(k,u,\Lambda,\gamma,Z){\models}C$.
Consequently, $t{\leq}s$ and $u{\in}R_{\L}(t)$.
Moreover, since $C{\in}\Gamma^{\circ}$, then by induction hypothesis, $u{\models}C$.
Hence, $s{\models}{\lozenge}C$.
From right to left, suppose $s{\models}{\lozenge}C$.
Since $({\mathcal T}_{\alpha},{\ll_{\alpha}},{\triangleright_{\alpha}})$ is clean, then there exists $(j,t,\Delta,\beta,Y){\in}W^{\prime}$ such that $(i,s,\Gamma,\alpha,X){R^{\prime}}(j,t,\Delta,\beta,Y)$ and $t{\models}C$.
Since $C{\in}\Gamma^{\circ}$, then by induction hypothesis, $(j,t,\Delta,\beta,Y){\models}C$.
Since $(i,s,\Gamma,\alpha,X){R^{\prime}}(j,t,\Delta,\beta,Y)$, then $(i,s,\Gamma,\alpha,X){\models}{\lozenge}C$.
\medskip
\end{proof}
\begin{lemma}\label{lemma:falsified:finite:model}
$A$ is falsified in $(W^{\prime},{\leq^{\prime}},{R^{\prime}})$.
\end{lemma}
\begin{proof}
By Lemma~\ref{lemma:truth:lemma:at:the:end:of:the:procedure} and the fact that $(W_{\L},{\leq_{\L}},{R_{\L}},V_{\L}),s_{0}{\not\models}A$.
\medskip
\end{proof}
\subsection*{The Finite Frame Property}
All in all, the desired result is within reach.
\begin{proposition}[Finite Frame Property]\label{proposition:fmp:fik}
\begin{itemize}
\item $\FIK{=}\Log({\mathcal C}_{\fcfra}){=}\Log({\mathcal C}_{\fucfra})
$\linebreak$
{=}\Log({\mathcal C}_{\fucfra}^{\omega})$,
\item $\LIK{=}\Log({\mathcal C}_{\fdcfra}){=}\Log({\mathcal C}_{\fducfra}){=}\Log({\mathcal C}_{\fducfra}^{\omega})$.
\end{itemize}
\end{proposition}
\begin{proof}
By Propositions~\ref{proposition:facile:about:inclusion} and~\ref{proposition:soundness:FIK:LIK} and Lemmas~\ref{lemma:frame:prime:is:appropriate} and~\ref{lemma:falsified:finite:model}.
\medskip
\end{proof}
\begin{corollary}
The membership problems in $\FIK$ and $\LIK$ are decidable.
\end{corollary}
\begin{proof}
By~\cite[Theorem~$6.13$]{Blackburn:et:al:2001} and Proposition~\ref{proposition:fmp:fik}.
\medskip
\end{proof}
\section{Conclusion}\label{section:conclusion}
Much remains to be done.
For example,
\begin{itemize}
\item adapt the above line of reasoning to the IMLs considered in~\cite{Balbiani:Gencer:preliminary:draft,Simpson:1994},
\item determine the computational complexity of the membership problem in these IMLs,
\item adapt the above line of reasoning to the intuitionistic variants of $\S4$ considered in~\cite{Balbiani:et:al:2021,Girlando:et:al:2023},
\item determine the computational complexity of the membership problem in these intuitionistic variants of $\S4$.
\end{itemize}
\section*{Acknowledgements}
We wish to thank Han Gao (Aix-Marseille University), Zhe Lin (Xiamen University), Nicola Olivetti (Aix-Marseille University) and Vladimir Sotirov (Bulgarian Academy of Sciences) for their valuable remarks.
\\
\\
Special acknowledgement is also granted to our colleagues of the Toulouse Institute of Computer Science Research for many stimulating discussions about the subject of this paper.
%
%
%
%
%
%
%
%
\bibliographystyle{named}
\end{document}